\RequirePackage{fix-cm}
\documentclass{svjour3}                     
\smartqed  
\usepackage{graphicx}
\usepackage{cite}
\usepackage{amsmath,amssymb,amsfonts}
\usepackage{algorithmic}
\usepackage{textcomp}
\usepackage{float}
\usepackage{subfigure}
\usepackage[colorlinks,linkcolor=blue,urlcolor=blue,citecolor=blue]{hyperref}
\begin{document}

\title{On Performance of Multiscale Sparse Fast Fourier Transform Algorithm}


\author{Bin Li$^{1}$, Zhikang Jiang$^{1}$ and Jie Chen$^{1}$}


\institute{Bin Li \at
              \email{sulibin@shu.edu.cn}           
           \and
           Zhikang Jiang \at
              \email{zkjiang@i.shu.edu.cn}  
           \and
           Jie Chen \at
              \email{jane.chen@shu.edu.cn}  
           \and
           $^{1}$ \quad School of Mechanical and Electrical Engineering and Automation, Shanghai University, Shanghai 200072, China
}

\date{Received: date / Accepted: date}

\maketitle

\begin{abstract}

Computing the Sparse Fast Fourier Transform(sFFT) of a $K$-sparse signal of size $N$ has emerged as a critical topic for a long time. The sFFT algorithms decrease the runtime and sampling complexity by taking advantage of the signal's inherent characteristics that a large number of signals are sparse in the frequency domain(e.g., sensors, video data, audio, medical image, etc.). The first stage of sFFT is frequency bucketization through one of these filters: Dirichlet kernel filter, flat filter, aliasing filter, etc. Compared to other sFFT algorithms, the sFFT algorithms using the flat filter is more convenient and efficient because the filtered signal is concentrated both in the time domain and frequency domain. Up to now, three sFFT algorithms sFFT1.0, sFFT2.0, sFFT3.0 algorithm have been proposed by the Massachusetts Institute of Technology(MIT) in 2013. Still, the sFFT4.0 algorithm using the multiscale approach method has not been implemented yet. This paper will discuss this algorithm comprehensively in theory and implement it in practice. It is proved that the performance of the sFFT4.0 algorithm depends on two parameters. The runtime and sampling complexity are in direct ratio to the multiscale parameter and in inverse ratio to the extension parameter. The robustness is in direct ratio to the extension parameter and in inverse ratio to the multiscale parameter. Compared with three similar algorithms or other four types of algorithms, the sFFT4.0 algorithm has excellent runtime and sampling complexity that ten to one hundred times better than the fftw algorithm, although the robustness of the algorithm is medium.


\end{abstract}

\section{Introduction}
\label{intro}
The widely popular algorithm to compute Discrete Fourier Transform (DFT) is the fast Fourier transform(FFT) invented by Cooley and Tukey, which can compute a signal of size $N$ in $O(N\text{log}N)$ time and use $O(N)$ samples. With the demand for low sampling ratio and big data computing, it motivates the new algorithms to replace the previous FFT algorithms that can compute DFT from a subset of the input data in sub-linear time. The new sFFT algorithms can reconstruct the spectrum with high accuracy by using only $K$ most significant frequencies. In terms of its excellent performance and generally satisfied assumptions, the technology of sFFT was named one of the ten Breakthrough Technologies in MIT Technology Review in 2012.

There are mainly two stages in the sFFT: frequency bucketization and spectrum reconstruction. Frequency bucketization is equivalent to hashing the frequency coefficients into $B(\approx  K)$ buckets through filters. The Dirichlet kernel filter only bins some frequency coefficients into one bucket one time. The aliasing filter looks like a comb is difficult to solve the worst case because there may be many frequency coefficients in the same bucket accidentally. The flat filter can be obtained by convoluted a Gaussian function with a box car window function. It is concentrated both in the time domain and frequency domain. After frequency bucketization, for the filtered signal, the sFFT algorithms need to reconstruct the spectrum by their unique method through their own framework. 

The sFFT algorithms using the Dirichlet kernel window filter is a randomized algorithm. The performance of the Ann Arbor fast Fourier transform(AAFFT0.5 \cite{IEEEexample:Gilbert2002Near}) algorithm was later improved in the AAFFT0.9 \cite{IEEEexample:Iwen2007Empirical}, \cite{IEEEexample:Gilbert2008A} algorithm through the use of unequally-spaced FFTs and binary search technique for spectrum reconstruction.
	
There are three frameworks for the sFFT algorithms using the aliasing window filter. The algorithms of the one-shot framework based on the compressed sensing solver are the so-called sFFT by downsampling in the time domain(sFFT-DT1.0 \cite{IEEEexample:Hsieh2013}, sFFT-DT2.0 \cite{IEEEexample:Hsieh2015}) algorithm. The algorithms of the peeling framework based on the bipartite graph are the so-called Fast Fourier Aliasing-based Sparse Transform(FFAST) \cite{IEEEexample:Pawar2013}, \cite{IEEEexample:Pawar2018} and R-FFAST \cite{IEEEexample:Pawar2015}, \cite{IEEEexample:Ong2019} algorithm. The algorithm of the iterative framework based on the binary tree search is the so-called Deterministic Sparse FFT(DSFFT \cite{IEEEexample:Plonka2018}) algorithm. Under the assumption of arbitrary sampling, the Gopher Fast Fourier Transform(GFFT) \cite{IEEEexample:Iwen2010}, \cite{IEEEexample:Iwen2013} algorithm and the Christlieb Lawlor Wang Sparse Fourier Transform(CLW-SFT) \cite{IEEEexample:LAWLOR2013}, \cite{IEEEexample:Christlieb2016} algorithm are aliasing-based search deterministic algorithm guided by the Chinese Remainder Theorem(CRT). The DMSFT \cite{IEEEexample:Merhi2019}(generated from GFFT) algorithm and CLW-DSFT \cite{IEEEexample:Merhi2019}(generated from CLW-SFT) algorithm use the multiscale error-correcting method to cope with noise.

There are two frameworks for the sFFT algorithms using the flat window filter. The algorithms of the one-shot framework are the sFFT1.0 \cite{IEEEexample:Hassanieh2012} and sFFT2.0 \cite{IEEEexample:Hassanieh2012} algorithm which can locate and estimate the $K$ largest coefficients in one shot by multiple random bucketization. The algorithms of the iterative framework are the sFFT3.0 \cite{IEEEexample:Hassanieh2012-2} and sFFT4.0 \cite{IEEEexample:Hassanieh2012-2} algorithm and etc. The sFFT3.0 algorithm can locate the position by using only two filtered signals inspired by the frequency offset estimation in the exactly sparse case. The sFFT4.0 algorithm introduced in this paper can locate the position block by block inspired by the multiscale frequency offset estimation in the generally sparse case. The new robust algorithm, so-called the Matrix Pencil FFT(MPFFT) \cite{IEEEexample:Chiu2014} algorithm, was proposed based on the sFFT3.0 algorithm. The paper \cite{IEEEexample:Li2020} summarizes the two frameworks and five reconstruction methods of these five corresponding algorithms.

As is shown in Table 1, the theoretical performance analysis of all the above algorithms can be seen. The paper\cite{IEEEexample:Gilbert2014} summarizes a three-step approach in the stage of spectrum reconstruction and provides a standard testing platform to evaluate different sFFT algorithms. There are also some researches try to conquer the sFFT problem from other aspects: complexity \cite{IEEEexample:Indyk2014}, \cite{IEEEexample:Kapralov2017}, performance \cite{IEEEexample:Chen2017}, \cite{IEEEexample:Lopez-Parrado2015}, software \cite{IEEEexample:Wang2016}, \cite{IEEEexample:Schumacher2014}, hardware \cite{IEEEexample:Abari2014}, higher dimensions \cite{IEEEexample:Kapralov2019}, \cite{IEEEexample:Wang2019}, implementation \cite{IEEEexample:Kumar2019}, \cite{IEEEexample:Pang2018} and special setting \cite{IEEEexample:Plonka2017}, \cite{IEEEexample:Plonka2016} perspectives.

This paper is structured as follows. Section 2 provides a brief overview of some notation and basic definitions that we will use in the sFFT. Section 3 introduces and analyzes the multiscale Sparse Fast Fourier Transform Algorithm in detail from five aspects: the overall flow of the algorithm, the steps of frequency reconstruction by the multiscale approach in one iteration, the performance of the algorithm in theory, the comparison with other algorithms. In Section 4, we do three categories of comparison experiments. The first is the experiments with different parameters. The second is to compare the algorithm with similar sFFT algorithms using the same flat filter. The third is to compare the algorithm with other types of sFFT algorithms. The analysis of the experiment results satisfies the inferences obtained in theory.

\section{Preliminaries}
In this section, we start with an overview of some notation and basic definitions that we will use in the sFFT.

\subsection{Notation}
The $N$-th root of unify is denoted by $\omega _{N}=e^{-2\pi \mathbf{i}/N}$. The DFT matrix of size $N$ is denoted by $\mathbf{F}_{N}\in \mathbb{C}^{N\times N}$ as follows:
\begin{equation}
\mathbf{F}_{N}[j,k]=\frac{1}{N}\omega _{N}^{jk}	\label{Eq1}
\end{equation}

The DFT of a vector $x\in \mathbb{C}^{N}$ (consider a signal of size $N$) is a vector $\hat{x}\in \mathbb{C}^{N}$ defined as follows:
\begin{equation}
\begin{split}
\hat{x}=\mathbf{F}_{N}x	\\
\hat{x}_{i}=\frac{1}{N} \sum_{j=0}^{N-1}x_{j}\omega _{N}^{ij}	
\end{split}	\label{Eq2}
\end{equation}

For $x_{-i}=x_{N-i}$, the convolution is defined as follows::
\begin{equation}
(x\ast y)_{i}=\sum_{j=0}^{N-1}x_{j}y_{i-j} \label{Eq3}
\end{equation}

For coordinate-wise product $(xy)_{i}=x_{i}y_{i}$, the DFT of $xy$ is performed as follows:
\begin{equation}
\widehat{xy}=\hat{x}\ast \hat{y} \label{Eq4}
\end{equation}

In the exactly sparse case, spectrum $\hat x$ is exactly $K$-sparse if it has exactly $K$ non-zero frequency coefficients while the remaining $N-K$ coefficients are zero. In the general sparse case, spectrum $\hat x$ is general $K$-sparse if it has $K$ significant frequency coefficients while the remaining $N-K$ coefficients are negligible. The goal of sFFT is to recover a $K$-sparse approximation $\hat {x'}$ by locating $K$ frequency positions $f_0, \dots, f_{K-1}$ and estimating $K$ largest frequency coefficients $\hat{x}_{f_{0}},\dots,\hat{x}_{f_{K-1}}$.

\subsection{The Technique of Random Spectrum Permutation}
	The first technique used in sFFT is the spectrum random permutation, including two operations: one is shift operation, another is scaling operation. The offset parameter is denoted by $\tau \in \mathbb{R}$. The matrix representing the shift operation is denoted by $\mathbf{S}_{\tau}\in \mathbb{R}^{N\times N} $ as follows:
\begin{equation}
\mathbf{S}_{\tau}[j,k]=\left\{\begin{matrix}  1, & j-\tau\equiv k(\text{mod} N) 
\\0,&\text{o.w.}
\end{matrix}\right.\label{Eq5}
\end{equation}

 The scaling parameter is denoted by $\sigma \in \mathbb{R}$. The matrix representing the scaling operation is denoted by $\mathbf{P}_{\sigma}\in \mathbb{R}^{N\times N} $ as follows:
\begin{equation}
\mathbf{P}_{\sigma}[j,k]=\left\{\begin{matrix}  1, &  \sigma {j} \equiv k(\text{mod} N)
\\0,&\text{o.w.} 
\end{matrix}\right.\label{Eq6}
\end{equation}

Suppose $\sigma^{-1} \in \mathbb{R}$ exists mod $N$, $\sigma^{-1}$ satisfies $\sigma^{-1}\sigma\equiv{1}(\text{mod}N)$. If a vector $x'\in \mathbb{C}^{N}, \ x'=\mathbf{S}_{\tau} \mathbf{P}_{\sigma} x$, such that:
\begin{equation}
\begin{split}
x'_i=x_{\sigma(i-\tau)}  \\
x'_{\sigma^{-1}i+\tau}=x_i
\end{split}\label{Eq7}
\end{equation}

According to the time shift property, the DFT of a random permutation signal is performed as follows: If $x'=\mathbf{S}_{\tau} \mathbf{P}_{\sigma} x$,  such that: 
\begin{equation}
\begin{split}
\hat{x'}_{\sigma i}=\hat{x}_{i}\omega ^{\sigma \tau i} \\
\hat{x'}_{i}=\hat{x}_{\sigma^{-1}i}\omega ^{\tau i}
\end{split}\label{Eq8}
\end{equation}

From Eq. (\ref{Eq8}), we can see the technique of random permutation isolates spectral components from each other.

\subsection{The Technique of Window Function}
The second technique used in sFFT is the window function which is an important mathematical tool that can be seen as a matrix multiply the original signal. We introduce two filters used in the sFFT algorithm mentioned in this paper.

The first filter is the flat window filter. We use a vector $G$ that is concentrated both in time and frequency domain, $G$ is zero except at a small number of time coordinates with supp$(G)\subseteq  [-w/2,w/2]$ and its Fourier Transform $\hat{G}$ is negligible except at a small fraction $L$ $(\approx {\varepsilon N})$ of the frequency coordinates (the pass region). The paper \cite{IEEEexample:Hassanieh2012} claim there exists a standard window function $G(\varepsilon, \ \varepsilon', \ \delta, \ w)$ satisfies Eq. (\ref{Eq9}). The filter can be obtained by convoluted a Gaussian function with a boxcar window function and supp$(G) = w = O(1/\varepsilon \ \text{log}(1/\delta)))$. By knowing these, we define filter $G \in \mathbb{C}^N$  be an $(L/N, \ L/2N, \ \delta, \ w)$ flat window. The width of the filter in the time domain is denoted by $w$, the width of the passband region in the frequency domain is denoted by $L$, the number of buckets is denoted by $B$ and $B = N/L$. 
\begin{equation}
\begin{aligned}
\left | \hat{G}_i  \right | \in [1-\delta ,1+\delta] \ for \ i\in [-\epsilon 'N,\epsilon 'N]\\
 \left |\hat{G}_i  \right |\in [0 ,\delta] \  for\  i\notin  [-\epsilon N,\epsilon N]\\
 \left | \hat{G}_i  \right | \in [0 ,1] \  for \left | i \right | \in[\epsilon 'N, \epsilon N ]
\end{aligned} \label{Eq9}
\end{equation}

The diagonal matrix whose diagonal entries represent filter coefficients in the time domain is denoted by matrix $\mathbf{Q}_{L}\in \mathbb{C}^{N\times N}$ as follows:
\begin{equation}
\mathbf{Q}_{L}[j,k]=\left\{\begin{matrix}  G_j, & j=k\\0,&\text{o.w.} 
\end{matrix}\right. \label{Eq10}
\end{equation}

The second filter is the frequency subsampled filter. Through the filter, the signal in the time domain is aliased, such that the corresponding signal in the frequency domain is subsampled. The matrix representing the aliasing operator is denoted by $\mathbf{U}_{L}\in \mathbb{R}^{B\times N}$ as follows:
\begin{equation}
\mathbf{U}_{L}[j,k]=\left\{\begin{matrix}  1, & j-k\equiv 0(\text{mod} B) 
\\0,&\text{o.w.}
\end{matrix}\right. \label{Eq11}
\end{equation}

The filtered signal obtained by the subsampled filter is denoted by $y_L=\mathbf{U}_{L}x$. According to the aliasing characteristic, the DFT of a filtered signal $\hat{y}_L=\mathbf{F}_{B}\mathbf{U}_{L}x$ is performed as follows:
\begin{equation}
\hat{y}_L[i]=\hat{x}[iL] \label{Eq12}
\end{equation}

\subsection{The Technique of Frequency Bucketization}
The process of frequency bucketization in this paper is achieved through two techniques mentioned above. It can be divided into the following three steps: random spectrum permutation, multiply flat window filter, Fourier transform. It can be equivalent to the signal multiply $\mathbf{F}_{B}\mathbf{U}_{L}\mathbf{Q}_{L}\mathbf{S}_{\tau} \mathbf{P}_{\sigma}$ and the filtered signal in each bucket is performed as follows: 
\begin{lemma} 
If $y_{L,\tau,\sigma}=\mathbf{U}_{L}\mathbf{Q}_{L}\mathbf{S}_{\tau} \mathbf{P}_{\sigma}x$, and $\hat{y}_{L,\tau,\sigma}=\mathbf{F}_{B}\mathbf{U}_{L}\mathbf{Q}_{L}\mathbf{S}_{\tau} \mathbf{P}_{\sigma}x$, such that: 
\begin{equation}
\begin{split}
\hat{y}_{L, \tau, \sigma}[0] \approx
 &\hat{G}_{\frac{L}{2}} \hat{x}_{\sigma^{-1}\left(-\frac{L}{2}\right)} \omega_{N}^{\tau\left(-\frac{L}{2}\right)}+\ldots \hat{G}_{-\frac{L}{2}+1} \hat{x}_{\sigma^{-1}(\frac{L}{2}-1)} \omega_{N}^{\tau (\frac{L}{2}-1)} \\
\hat{y}_{L, \tau, \sigma}[1] \approx &\hat{G}_{\frac{L}{2}} \hat{x}_{\sigma^{-1}\left(\frac{L}{2}\right)} \omega_{N}^{\tau\left(\frac{L}{2}\right)}+\ldots \hat{G}_{-\frac{L}{2}+1} \hat{x}_{\sigma^{-1}(\frac{3L}{2}-1)} \omega_{N}^{\tau (\frac{3L}{2}-1)}\\
\hat{y}_{L, \tau, \sigma}[i] \approx &\hat{G}_{\frac{L}{2}} \hat{x}_{\sigma^{-1}\left(\frac{(2i-1)L)}{2}\right)} \omega_{N}^{\tau\left(\frac{(2i-1)L}{2}\right)}+\ldots \hat{G}_{-\frac{L}{2}+1} \hat{x}_{\sigma^{-1}(\frac{(2i+1)L}{2}-1)} \omega_{N}^{\tau (\frac{(2i+1)L}{2}-1)}
\end{split} \label{Eq13}
\end{equation}
\end{lemma} 
\begin{proof}
\begin{equation*}
\begin{split} 
&\text{(P0)} x'=\mathbf{S}_\tau \mathbf{P}_\sigma  x   \Rightarrow \hat{x'}_{i}=\hat{x}_{\sigma^{-1}i}\omega ^{\tau i} \\
&\text{(P1)} x''=  \mathbf{Q}_L   x'    \Rightarrow \begin{bmatrix} \hat{x''}[0] \\ \cdots \\ \hat{x''}[L]\\ \cdots \\ \hat{x''}[iL] \\ \cdots \end{bmatrix} = \begin{bmatrix} \hat{x'}[0] \\ \cdots \\ \hat{x'}[L]\\ \cdots \\ \hat{x'}[iL]
\\ \cdots \end{bmatrix}  \ast \begin{bmatrix} \hat{G}[0] \\ \cdots \\ \hat{G}[L]\\ \cdots \\ \hat{G}[iL] \\ \cdots \end{bmatrix} \\ &\text{(P2)} \left | {G}[i] \right | =\left\{\begin{matrix}  \approx 1, & i\in [-\frac{L}{2}+1, \frac{L}{2}] \\ \approx  0,&\text{o.w.} \end{matrix}\right. \\ &\text{(P3)}
\hat{y}_{L,\tau,\sigma} =\mathbf{F}_{B}\mathbf{U}_{L}x''\Rightarrow \hat{y}_{L,\tau,\sigma}[i]=\hat{x''}[iL] \\ &\text{(P4)} \hat{y}_{L,\tau,\sigma}=\mathbf{F}_{B}\mathbf{U}_{L}\mathbf{Q}_{L}\mathbf{S}_{\tau} \mathbf{P}_{\sigma}x \\ &\text{Based on the above-mentioned properties we get Eq. (13)} 
\end{split} 
\end{equation*}
\end{proof}

If the set $I$ is a set of coordinates position, the position $f=(\sigma^{-1}u) \text{mod} N \in I$, suppose there is no hash collision in the bucket $i$, $i$ = round($u/L$), round() means to make decimals rounded. Through Eq. (\ref{Eq13}), we can get Eq. (\ref{Eq14}).
\begin{equation}
\begin{split}
\hat{y}_{L, \tau, \sigma}[i] \approx \hat{G}_{iL-u} \hat{x}_{\sigma^{-1}u} \omega_{N}^{\tau u} \ \text{for} \ u\in [\frac{(2i-1)L}{2},\frac{(2i+1)L}{2}-1]\\
\hat{x}_{f} \approx \hat{y}_{L, \tau, \sigma}[i]\omega_{N}^{-\tau u}/\hat{G}_{iL-u}  \ \text{for} \ u=\sigma f \text{mod} N,i=\text{round}(u/L)
\end{split} \label{Eq14}
\end{equation}

In all, the performance analysis of frequency bucketization is described as follows: random spectrum permutation($x'=\mathbf{S}_\tau \mathbf{P}_\sigma  x$, it cost 0 runtime), flat window filter($x''=\mathbf{Q}_L x'$, it cost $w$ runtime and $w$ samples), Fourier Transform of the aliasing signal($\hat{y}_{L,\tau,\sigma}=\mathbf{F}_{B}\mathbf{U}_{L}x''$, it cost $B$log$B$ runtime and 0 samples). So frequency bucketization one round cost $w+B$log$B$ runtime and $w$ samples.

\section{Algorithms analysis}
As mentioned above, frequency bucketization can decrease runtime and sampling complexity in the advantage of all operations are calculated in $B$ dimensions($B = O(K), B << N$). After frequency bucketization; the filtered signal $\hat{y}_{L,\tau,\sigma}$ can be obtained by original signal $x$; the subsequent work is spectrum reconstruction by identifying frequencies that are isolated in their buckets. Suppose in one bucket, the number of significant frequencies is denoted by $p$. In most cases, $p = 0$ respects sparsity. In a small number of cases, $p = 1$ respects only one significant frequency in the bucket. Only in very little cases, $p\geq  2$ respects more than one frequencies aliasing in the bucket. 

There are two frameworks to reconstruct spectrum: the one-shot framework and the iterative framework. The solver of the one-shot framework is the probability and statistics voting method based on multiple random bucketization. The typical algorithms of the one-shot framework are the sFFT1.0 algorithm and the sFFT2.0 algorithm. The multiscale sFFT algorithm named the sFFT4.0 algorithm introduced in this paper adopts the iterative framework. There are two improvements to the iterative framework. The first improvement is that once the frequency coefficients were found and estimated, they can be subtracted from the signal. It can reduce the amount of work to be done in subsequent steps. It is not necessary to update the whole input signal. Instead, it is sufficient to update the $B$-dimensional buckets, so the removal of the effects of already found coefficients can be done in $O(B)$ time. The second advantage is an improved method for finding the signal's significant frequency coordinates rather than the voting method. In the iterative algorithms, $R$(= two or $\text{log}_l L$) rounds is enough in their own ways. But in the one-shot framework, $R(\approx logN)$ rounds must be required to get correct locations at a high probability.

\subsection{The overall process of the multiscale sFFT algorithm}
The operation in each iteration can recover the single frequency isolated in the bucket. It means that each spectrum reconstruction can solve the case of buckets with $p<=1$. The idea based on the iterative framework is that the filtered signal subtracts these estimated frequencies in the next iteration must be more sparse and easier to recover. Moreover, the initial aliasing frequencies will be separated in subsequent iterations.

In the No.$m$' iteration, let $K_m$ be the expected sparsity($K_1=K, K_2=K/2,\dots$), $R_m$ be how many rounds in the No.$m$’ iteration($R_m$ is equal to two for the sFFT3.0 algorithm, $R_m$ is approximately equal to $\text{log}_l L_m$ for the sFFT4.0 algorithm and $l$ is the multiscale parameter respecting the number of blocks), $B_m$ be the number of buckets, $L_m$ be the size of one bucket, $w_m$ be the support of flat filter $G$, $\hat{y}_{L,\tau,\sigma}$ be filtered spectrum, $\hat{y}_{\text{update}}$ be filtered spectrum have already known, $\hat{y'}_{L,\tau,\sigma}$ be the spectrum need to recover($ \hat{y'}_{L,\tau,\sigma}  = \hat{y}_{L,\tau,\sigma} - \hat{y}_{\text{update}}$), $\hat{x}^{m-1}$ be the last result, $\hat{x'}^{m}$ be the recovered spectrum, $\hat{x}^{m}$ be the new result $( \hat{x}^{m} = \hat{x}^{m-1} + \hat{x'}^{m})$, set $ \tau = \{\tau _1 \dots \tau _R\}$ be offset parameter, set $\sigma = \{\sigma_1 \dots \sigma _R\}$ be scaling parameter. The system block diagram of the first iteration and second iteration is shown in Fig. \ref{fig1} and Fig. \ref{fig2}.
\begin{figure}[H]
\centering
\includegraphics[width=12 cm]{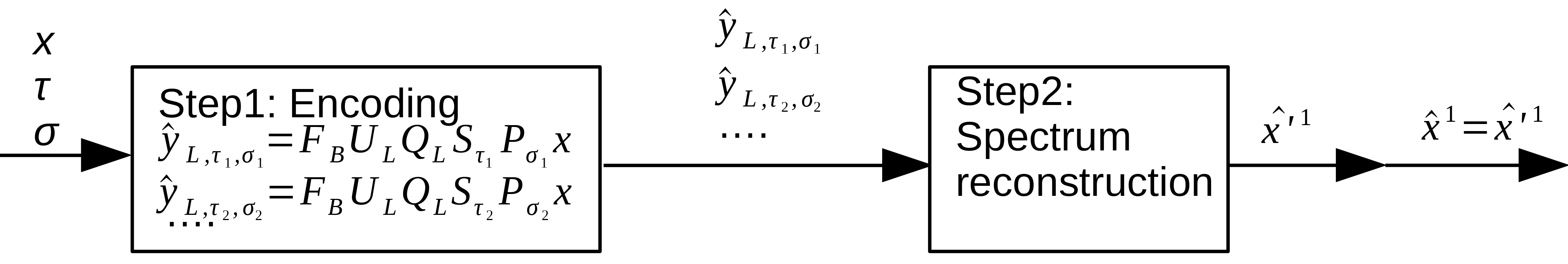}
\caption{The system block diagram of the first iteration.}\label{fig1}
\end{figure}   
\begin{figure}[H]
\centering
\includegraphics[width=12 cm]{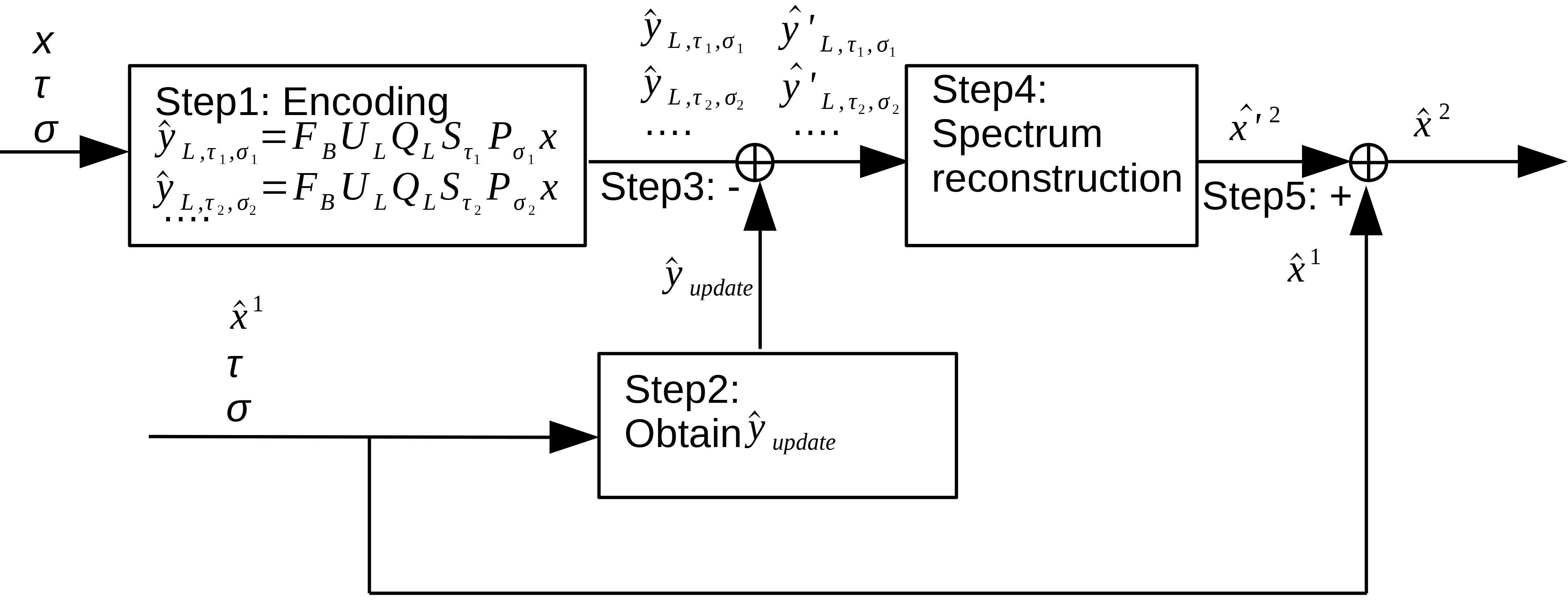}
\caption{The system block diagram of the second iteration.}\label{fig2}
\end{figure} 

The first iteration is divided into the following two steps: Step 1 Encoding: Run $R_1$ bucketization rounds for $K_1, B_1, L_1, \{\tau _1 \dots \tau _R\},  \{\sigma_1 \dots \sigma _R\}$ to calculate $\hat{y}_{L,\tau,\sigma}=\mathbf{F}_{B}\mathbf{U}_{L}\mathbf{Q}_{L}\mathbf{S}_{\tau} \mathbf{P}_{\sigma}x$ representing filtered spectrum. It costs $R_1(w_1+B_1\text{log}B_1)$ runtime and $R_1w_1$ samples. Step 2 Spectrum reconstruction: Recover the spectrum $\hat{x'}^{1}$ of $\hat{y}_{L,\tau,\sigma}$ by the multiscale method. It costs $T_1$ runtime($T_m$ is denoted by the runtime complexity of the spectrum reconstruction in the No.$m$' iteration. The detail will be explained in the following chapters). The result of the first iteration is equal to the recovered spectrum($\hat{x}^{1} = \hat{x'}^{1}$). It costs $R_1(w_1+B_1\text{log}B_1)+T_1$ runtime and $R_1w_1$ samples in the first iteration.

The second iteration is divided into the following five steps: Step 1 Encoding: Run $R_2$ bucketization rounds for $K_2, B_2, L_2, \{\tau _1 \dots \tau _R\},  \{\sigma_1 \dots \sigma _R\}$ to calculate $\hat{y}_{L,\tau,\sigma}=\mathbf{F}_{B}\mathbf{U}_{L}\mathbf{Q}_{L}\mathbf{S}_{\tau} \mathbf{P}_{\sigma}x$ representing filtered spectrum. It costs $R_2(w_2+B_2\text{log}B_2)$ runtime and $R_2w_2$ samples. Step 2 Obtain filtered spectrum have already known: Run $R_2$ times rounds for $\hat{x}^{1}$, set $\tau$, set $\sigma$ and Eq. (\ref{Eq14}) to obtain $\hat{y}_{\text{update}}$ representing filtered spectrum have already known by the last result $\hat{x}^{1}$. The runtime complexity depends on the support of $\hat{x}^{1}$, so it costs $R_2(K-K_2)$ runtime. Step 3 Obtain spectrum need to recover: $ \hat{y'}_{L,\tau,\sigma}  = \hat{y}_{L,\tau,\sigma} - \hat{y}_{\text{update}}$; representing spectrum need to recover. It costs $R_2(K-K_2)$ runtime. Step 4 Spectrum reconstruction: Recover the spectrum $\hat{x'}^{2}$ of $\hat{y'}_{L,\tau,\sigma}$ by the multiscale method. It costs $T_2$ runtime. Step 5 Obtain the new result: $ \hat{x}^{2} = \hat{x}^{1} + \hat{x'}^{1}$; representing the result of the second iteration. It costs $K_2$ runtime. It costs $R_2(w_2+B_2\text{log}B_2+K-K_2)+T_2$ runtime and $R_2w_2$ samples in the second iteration. 

The subsequent iteration is very similar to the second iteration. It costs $R_m(w_m+B_m\text{log}B_m+K-K_m)+T_m$ runtime and $R_mw_m$ samples in the No.$m$' iteration. If it is the last iteration, the final result is $\hat{x}^{m}$. Otherwise, $\hat{x}^{m}$ will be the input to make $\hat{y}_{\text{update}}$ for the next iteration. 

\subsection{The process of spectrum reconstruction}
The method of spectrum reconstruction, which is based on the multiscale approach, can locate the position block by block. Next, we will introduce the steps, principles and performance of this method in theory in detail.

For convenience, take bucket $i$ as an example; other buckets are similar. In bucket $i$, suppose there is only one significant frequency, the position of this single frequency in bucket $i$ is denoted by $f={\sigma^{-1}u} \ \text{for} \ u\in [\frac{(2i-1)L}{2},\frac{(2i+1)L}{2}-1]$, the value of this frequency is denoted by $\hat{x}_f=\hat{x}_{\sigma^{-1}u} $. If the bucket is an aliasing bucket, the aliasing frequencies will be separated in subsequent iterations by using different $\sigma$.

The parameters used for the first multiscale phase location are as follows: let $L_1$ be the length in this location($L_1 =L= N/B$), $l$ be the multiscale parameter respecting the number of blocks, $r_1$ be the size of one block($r_1=L_1/l$), $u$ be the real position, and $u'$ be the located position which calculated by two rounds, $u_{1\text{min}}$ be the starting position in the region($u_{1\text{min}}= (2i-1)L/2$), $u_{1\text{max}}$ be the termination position in the region($u_{1\text{max}}= (2i+1)L/2$), the range of the region is $ [\frac{(2i-1)L}{2},\frac{(2i+1)L}{2}-1]$ or $[u_{1\text{min}},u_{1\text{max}} )$ or $[u_{1\text{min}},u_{1\text{min}}+r_1 l)$, As is shown in Fig. \ref{fig4}a, the region can divided into $l$ blocks as follows: $[u_{1\text{min}},u_{1\text{min}}+r_1), [u_{1\text{min}}+r_1,u_{1\text{min}}+2r_1), …, [u_{1\text{min}}+(l-1)r_1,u_{1\text{min}}+lr_1)$. During the first location, $\tau_1=0,\sigma_1=\sigma$ for the first round and $\tau_2=\tau,\sigma_2=\sigma$($\tau\approx N/L_1$) for the second round. We use $\Phi (\theta )$ to denote the phase of $\theta$; function $\Phi (\theta )$ satisfies $e^{\Phi (\theta )\mathbf{i}}=\theta, \ \Phi (\theta ) \in [0,2\pi)$. In bucket $i$, the noise for the first round and for the second round is defined as $S_{\tau_1}[i]$ and $S_{\tau_2}[i]$ respectively, they satisfy Eq. (\ref{Eq15}). Through Eq. (\ref{Eq15}), we can obtain Eq. (\ref{Eq16}) and Eq. (\ref{Eq17}) where $\left |   \! \Phi (\frac{\hat{y}_{\tau_1}[i]- S_{\tau_1}[i]}{\hat{y}_{\tau_2}[i]- S_{\tau_2}[i]}) - \Phi (\frac{\hat{y}_{\tau_1}[i]}{\hat{y}_{\tau_2}[i]})\right |$ is denoted by $\Delta \Phi (\theta )$ and $\left | u-u' \right |$ is denoted by $\Delta u$. 
\begin{equation}
\begin{split}
S_{\tau_1}[i]=\hat{y}_{L_1, \tau_1, \sigma}[i] - \hat{G}_{iL-u} \hat{x}_{\sigma^{-1}u}\\
S_{\tau_2}[i]=\hat{y}_{L_1, \tau_2, \sigma}[i] - \hat{G}_{iL-u} \hat{x}_{\sigma^{-1}u}\omega^{\tau u}_N
\end{split} \label{Eq15}
\end{equation}

\begin{equation}
\begin{split}
\omega_{N}^{\tau u}\! =\! \frac{\hat{y}_{\tau_1}[i] - S_{\tau_1}[i]}{\hat{y}_{\tau_2}[i]  - S_{\tau_2}[i]} \!
 \Rightarrow \! u \tau\text{mod}N \frac{2\pi }{N} \!= \! \Phi (\frac{\hat{y}_{\tau_1}[i]- S_{\tau_1}[i]}{\hat{y}_{\tau_2}[i]- S_{\tau_2}[i]}) \\
\omega_{N}^{\tau u'}= \frac{\hat{y}_{\tau_1}[i]}{\hat{y}_{\tau_2}[i]} \Rightarrow u' \tau \text{mod}N \frac{2\pi }{N}=\Phi (\frac{\hat{y}_{\tau_1}[i]}{\hat{y}_{\tau_2}[i]})
\end{split} \label{Eq16}
\end{equation}

\begin{equation}
\Delta \Phi (\theta )=\Delta u \cdot \tau\cdot \frac{2\pi}{N} \approx \Delta u\cdot \frac{2\pi}{L_1} \label{Eq17}
\end{equation}

There are two possibilities for the relationship between $u$ and the block where $u'$ is located. The first possibility is the ideal case, which is shown in Fig. \ref{fig3}a, where $u$ is contained in the located block. In this case, the Eq. $\Delta u \leq \frac{r_1}{2}$ need to be satisfied, so that we can get Eq. (\ref{Eq18}). The second possibility is the non-ideal case, which is shown in Fig. \ref{fig3}b, where $u$ is not contained in the located block. The block needs to be expanded to a new region to contain $u$. The extension parameter is defined as $q$ respecting new addition block is $q$ times of the original block. For example, $q=1$ means the new region is 0.5 times larger than the original block in the upper and lower bounds of the original block; and becomes as large as 2 times of the original block. In another example, $q=0.5$ means the new region becomes as large as 1.5 times of the original block. In this case, the Eq. $\frac{q}{2}r_1\geq  \Delta u $ need to be satisfied, so that we can get Eq. (\ref{Eq19}). From Eq. (\ref{Eq18}), we can see the upper limit of the multiscale parameter $l$ is in inverse ratio to $\text{max}\left \{ \Delta \Phi  (\theta )\right \}$. From Eq. (\ref{Eq19}), we can see the lower limit of the extension parameter $q$ is in direct ratio to $\text{max}\left \{ \Delta \Phi  (\theta )\right \}$ and $l$.  These two parameters do not depend on the number of multiscale location.
\begin{equation}
l\leq \frac{\pi}{\text{max} \left \{\Delta \Phi  (\theta )\right \}} \label{Eq18}
\end{equation}
\begin{equation}
q\geq \text{max}\left \{ \Delta \Phi  (\theta )\right \} \cdot \frac{l}{\pi}  \label{Eq19}
\end{equation}
\begin{figure} [H]\centering  
\subfigure[The ideal case of the real position contained in the located block. ] { 
\includegraphics[width=0.45\columnwidth]{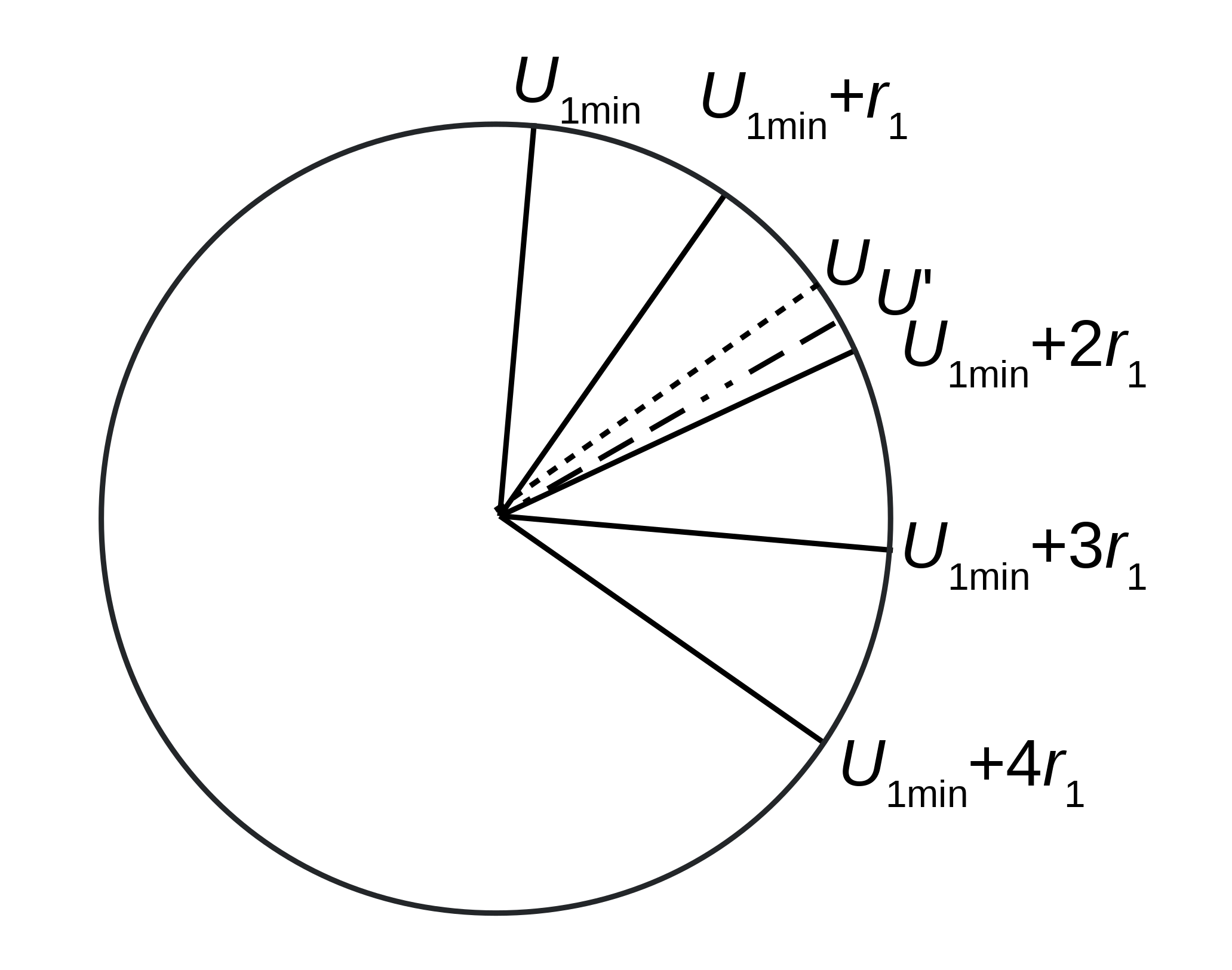}
}   
\subfigure[The non-ideal case of the real position uncontained in the located block.] { 
\includegraphics[width=0.45\columnwidth]{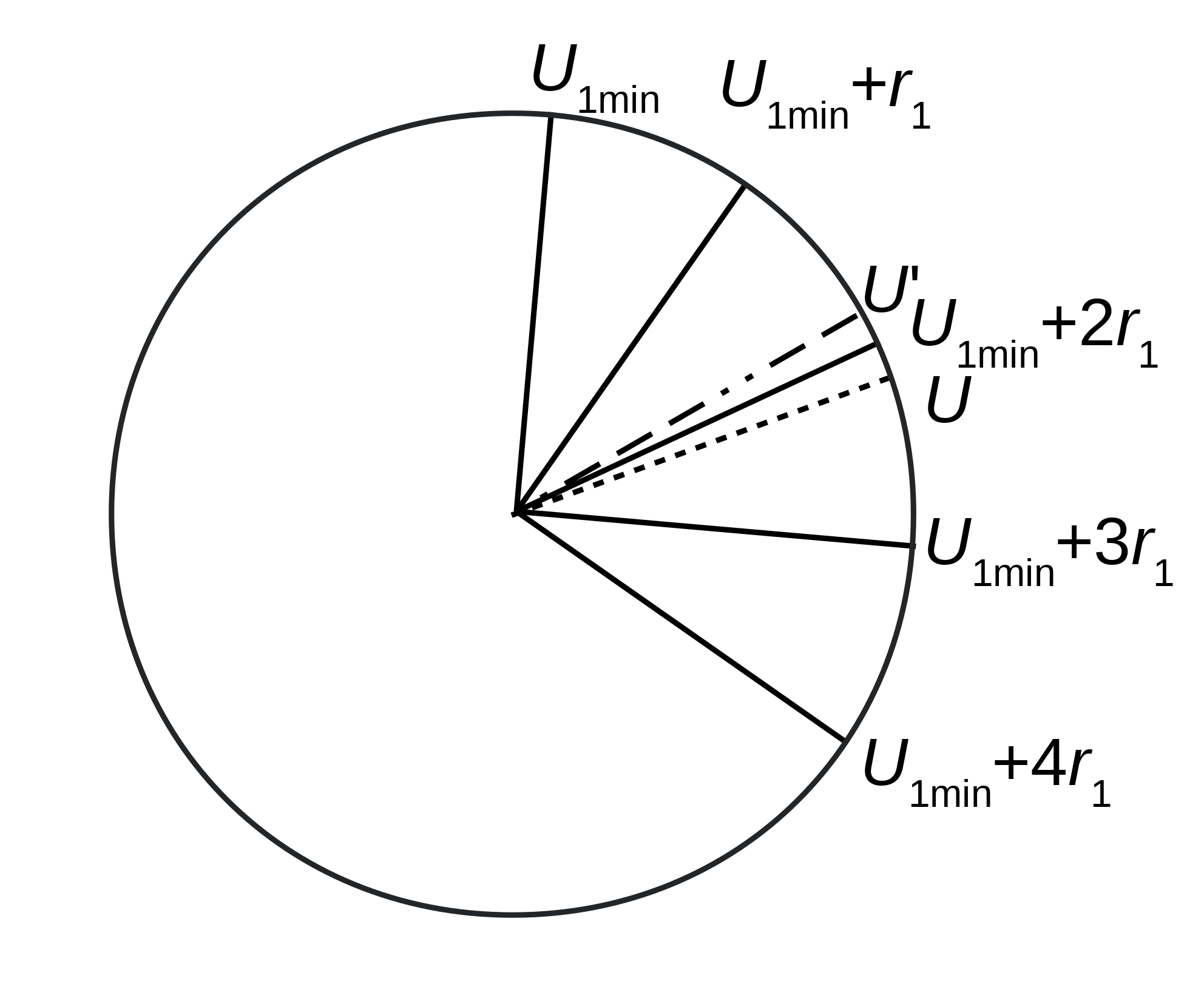}   
}   
\caption{Two possibilities for the relationship between the real position and the located block.}  \label{fig3} 
\end{figure}

After knowing these assumptions, we can explain the process of the first location clearly. As is shown in Fig. \ref{fig4}a, initially, the range of the old region is $[u_{1\text{min}},u_{1\text{max}} )$. As is shown in Fig. \ref{fig4}b, through scale operation, we enlarge the length of the old region from $L_1$ to almost $N$. As is shown in Fig. \ref{fig4}c, we can determine which block is we purchased according to $u'$(green space) and expend to a new region(yellow space). The number of the located block is denoted by $l_r$ which satisfied the Eq. $\text{argmin}\left |(u_{1\text{min}}+(l_r+0.5) r_1)\tau \text{mod}N \frac{2\pi }{N}-\Phi (\frac{\hat{y}_{\tau_1}[i]}{\hat{y}_{\tau_2}[i]})\right |$. After obtain the located block $[u_{1\text{min}}+l_r r_1, u_{1\text{min}}+(l_r+1)r_1]$, we determine the new region just as $[u_{1\text{min}}+(l_r-q/2)r_1, u_{1\text{min}}+(l_r+1+q/2)r_1]$. As is shown in Fig. \ref{fig4}d, we can do the next multiscale approach location through the new region. For the second location, the length of the region is changed from $L_1$ to $L2=L_1 (q+1)/l$, the starting position $u_{2\text{min}}= u_{1\text{min}}+(l_r-q/2)r_1$, the termination position $u_{2\text{max}}=  u_{1\text{min}}+(l_r+1+q/2)r_1$. It can be seen the speed of approach is in direct ratio to $l$ and in inverse ratio to $q$.

\begin{figure} [H]\centering  
\subfigure[The range of the old region.] { 
\includegraphics[width=0.4\columnwidth]{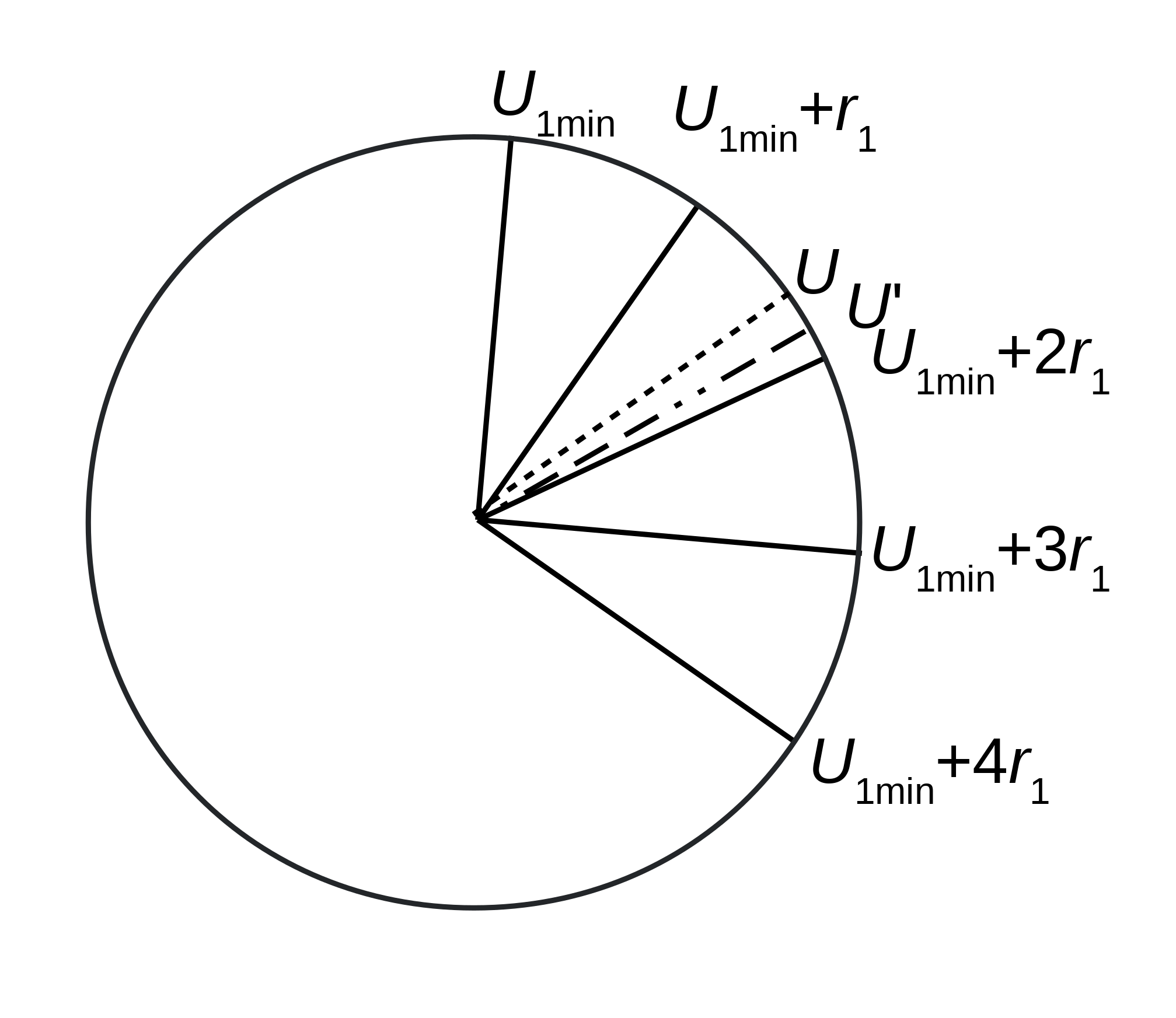}
}   
\subfigure[The range of the expanded old region.] { 
\includegraphics[width=0.5\columnwidth]{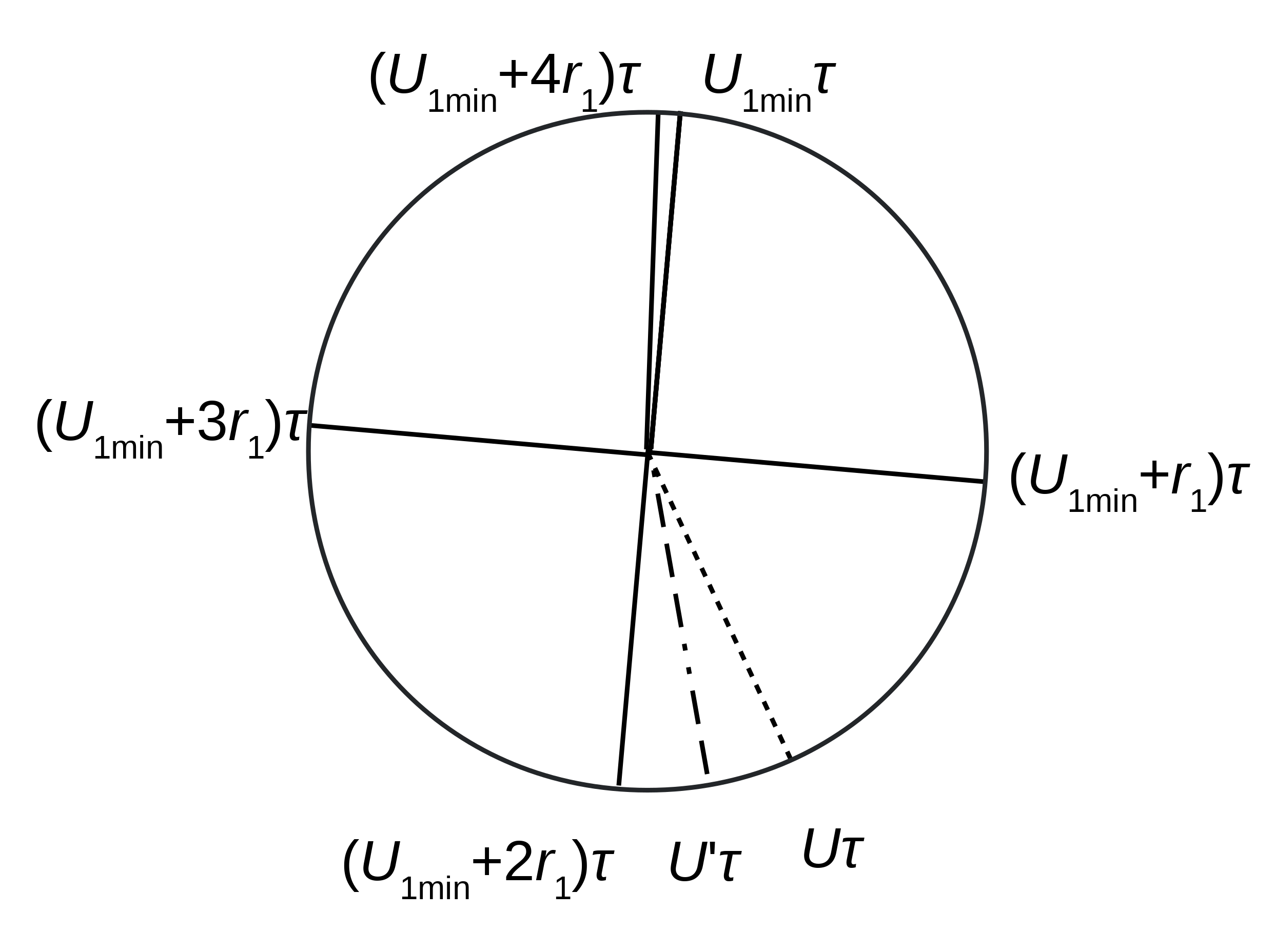}   
}   

\subfigure[Determine and expand the block.] { 
\includegraphics[width=0.5\columnwidth]{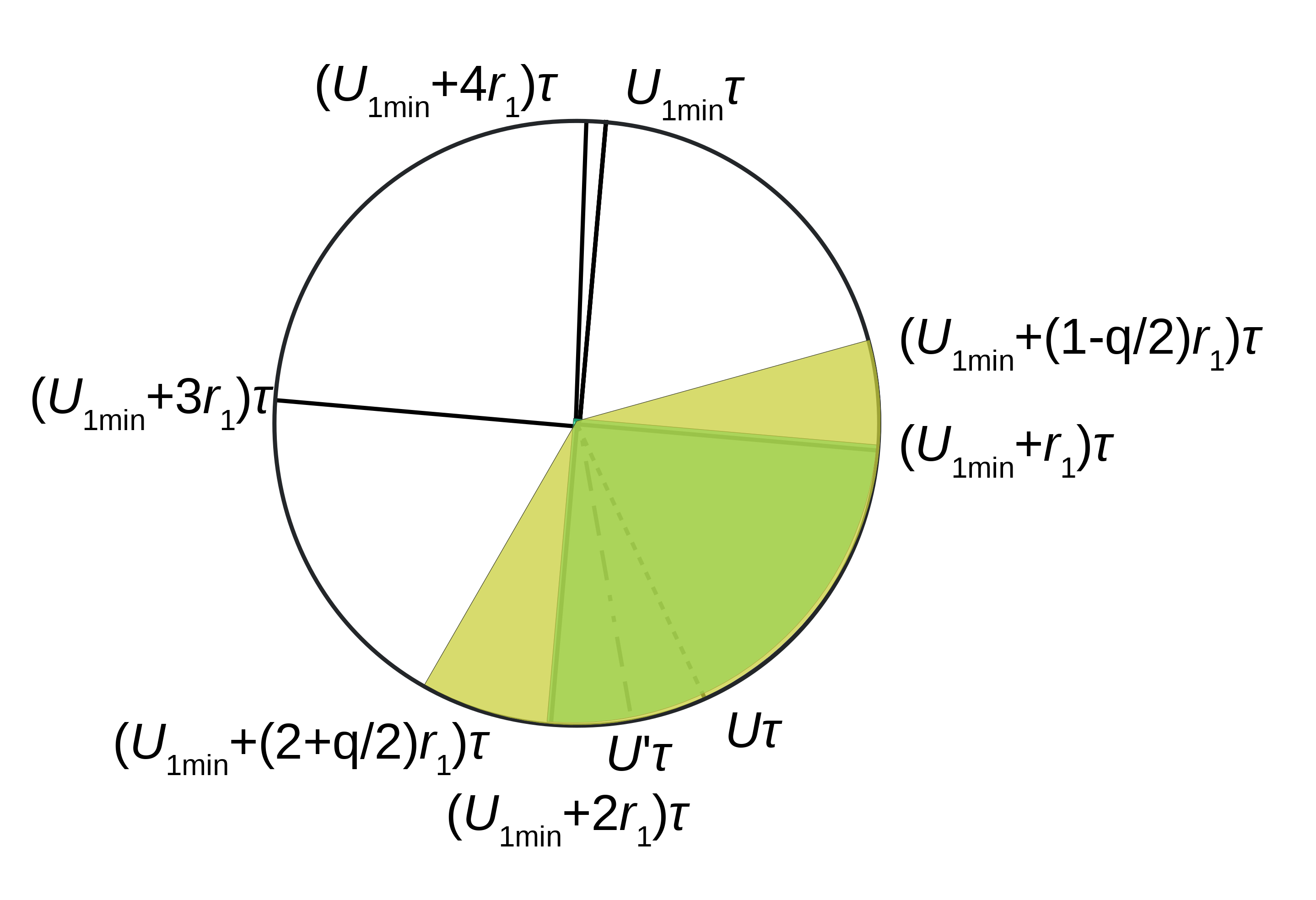}
}   
\subfigure[Restore and determine the new region.] { 
\includegraphics[width=0.35\columnwidth]{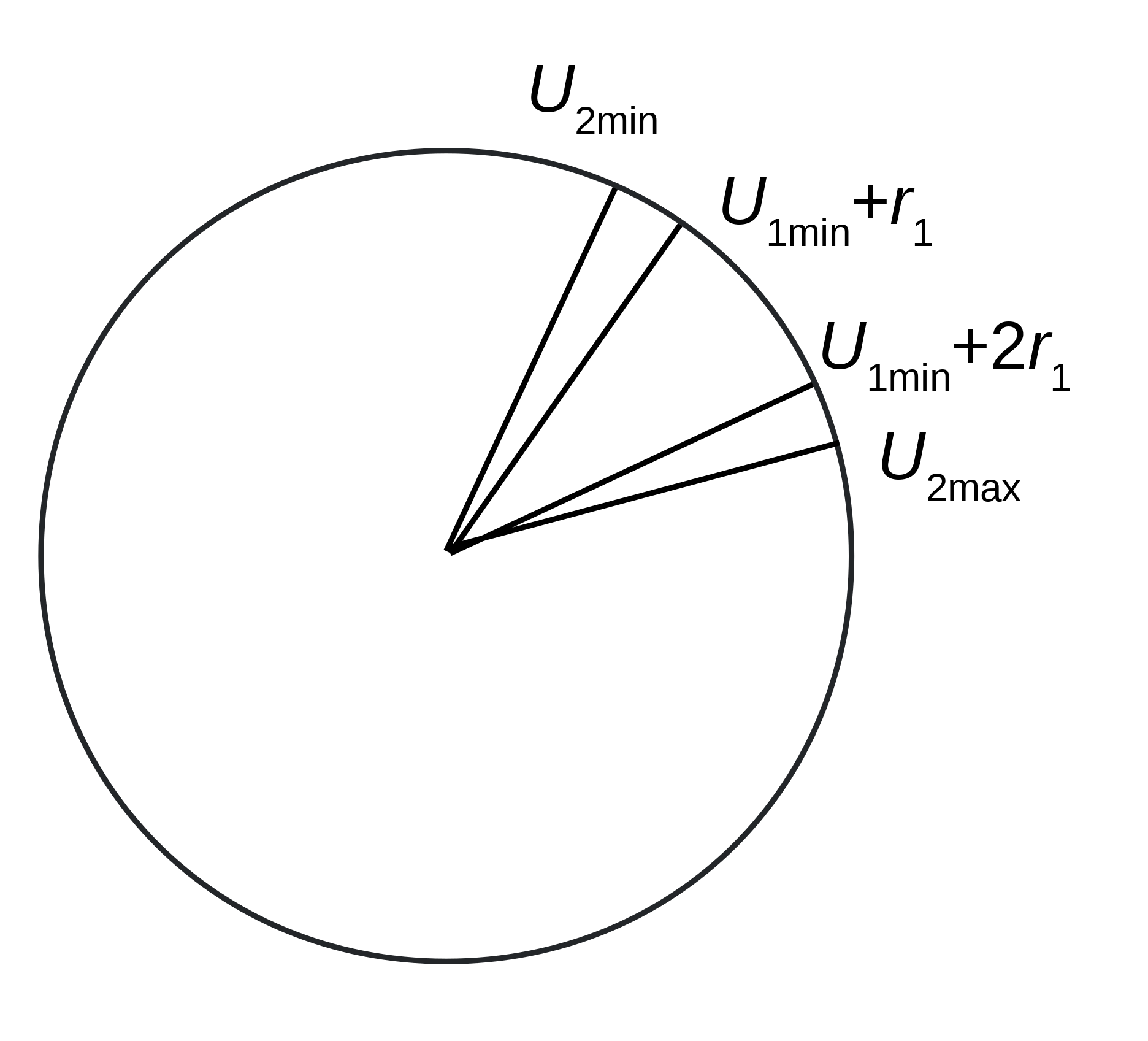}   
}   
\caption{The process of the first location.}  \label{fig4} 
\end{figure}

In the second multiscale location, $L_2, u_{2\text{min}}, u_{2\text{max}}$ have been known by the first location, $r_2$ be the size of one block($r_2=L_2/l$), the range of the region is $[u_{2\text{min}},u_{2\text{max}})$ or $[u_{2\text{min}},u_{2\text{min}}+r_2l]$, it can divided into $l$ blocks as follows: $[u_{2\text{min}},u_{2\text{min}}+r_2), [u_{2\text{min}}+r_2,u_{2\text{min}}+2r_2), …, [u_{2\text{min}}+(l-1)r_2,u_{2\text{min}}+lr_2)$. By using the same way, we can obtain the new region $[u_{3\text{min}},u_{3\text{max}})$. For the third location, the length of the region is changed from $L_2$ to $L_3=L_2(q+1)/l$. By continuously reducing the location region, we can locate the position block by block. It can be proved that the number of times needed to locate is $\text{log}_{(l/(q+1))}L$.

From Eq. (\ref{Eq18}) and Eq. (\ref{Eq19}), we can see $\text{max}\left \{ \Delta \Phi  (\theta )\right \}$ is very important to determine $l$ and $q$. By its definition, it is mainly determined by the noise $S_{\tau}[i]$. The noise $S_{\tau}[i]$ is mainly determined by the length of buckets $L$ and the signal noise ratio(SNR). If the SNR is low, the noise must be increased, and if the length one bucket is small, the noise must be increased too. So we can do a Monte Carlo experiment to prove it. We calculate the logarithm of the error of phase $\log_{10}(\Delta \Phi (\theta ))$ and computing the PDF(Probability Distribution Function) of the value $\log_{10}(\Delta \Phi (\theta ))$ by the input signals with the same $K$ of different $N$ under different SNR circumstances only if one significant frequency in the bucket, then we obtain Fig. \ref{fig5}. From Fig. \ref{fig5}, the PDF presents a normal distribution, and we can see with the development of SNR(from the yellow space to the green space) or the development of $L$(from Fig. \ref{fig5}a to Fig. \ref{fig5}d), the probability of small $ \Delta \Phi  (\theta )$ increases. If we want to keep the probability greater than 0.99 under the condition of SNR = -20(yellow space), $\text{max}\left \{ \Delta \Phi  (\theta )\right \}$ is about $10^{0.5} \approx  3.2$ to about $10^0 \approx  1$ in the case of different $L$. The upper limit of $l$ is approximately equal to $\pi /3.2\approx 1$ or $\pi /1\approx 3$ through Eq. (\ref{Eq18}) respectively. Under the condition of SNR = 0(red space), $\text{max}\left \{ \Delta \Phi  (\theta )\right \}$ is about $10^0 \approx  1$ to $10^{-0.5} \approx  0.3$ in the case of different $L$. The upper limit of $l$ is approximately equal to $\pi /1\approx 3$ or $\pi /0.3\approx 9$ respectively. Under the condition of SNR$\geq 20$(left of the purple space), $\text{max}\left \{ \Delta \Phi  (\theta )\right \}$ is about $10^{-0.5} \approx  0.3$ to $10^{-1} \approx  0.1$ in the case of different $L$. The upper limit of $l$ is approximately equal to $\pi /0.3\approx 9$ or $\pi /0.1\approx 30$ respectively. As to the lower limit of $q$, if $l$ is equal to the upper limit, the lower limit of $q$ is equal to 1 through the Eq. $q\geq \text{max}\left \{ \Delta \Phi  (\theta )\right \} \cdot \frac{l}{\pi} = {\text{max} \left \{\Delta \Phi  (\theta )\right \}} \cdot \frac{\pi}{\text{max} \left \{\Delta \Phi  (\theta )\right \}}\cdot\frac{1}{\pi}=1$, and if we use small $l$, we can get small $q$ as well. For example, under the condition of SNR = 0, $\text{max}\left \{ \Delta \Phi  (\theta )\right \}$ is about 0.3, and the upper limit of $l$ is approximately equal to 9 with big $L$. If we choose $l$ is equal to 4, the lower limit of $q$ satisfies $q \geq 0.3 \cdot 4 / \pi \approx 0.4$ through Eq. (\ref{Eq19})(Remarks: It is easy to know the PDF of the error of phase $ \Delta \Phi  (\theta )$ does not change much with different $\tau$ or different $\sigma$). 

\begin{figure} [H]\centering  
\subfigure[In the case of $K$ = 50, $N$ = 8192, $L$=32.] { 
\includegraphics[width=0.45\columnwidth]{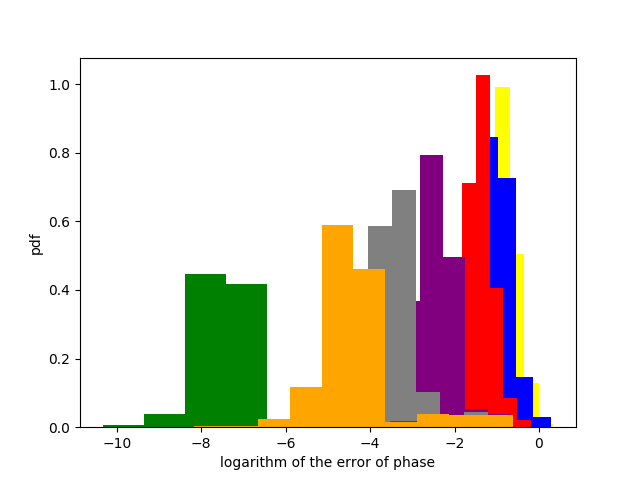}
}  
\subfigure[In the case of $K$ = 50, $N$ = 131072, $L$=2048.] { 
\includegraphics[width=0.45\columnwidth]{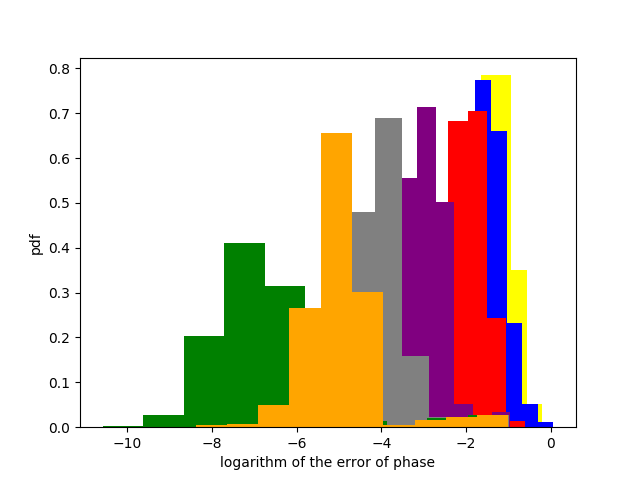}   
} 

\subfigure[In the case of $K$ = 50, $N$ = 1048576, $L$=2048.] { 
\includegraphics[width=0.45\columnwidth]{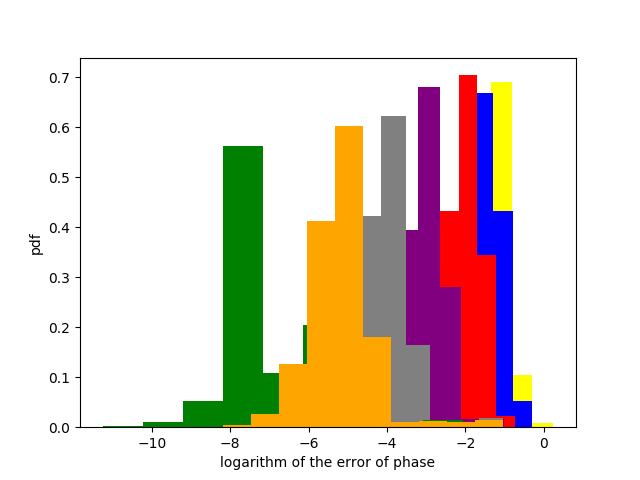}   
}  
\subfigure[In the case of $K$ = 50, $N$ = 4194304, $L$=8192.] { 
\includegraphics[width=0.45\columnwidth]{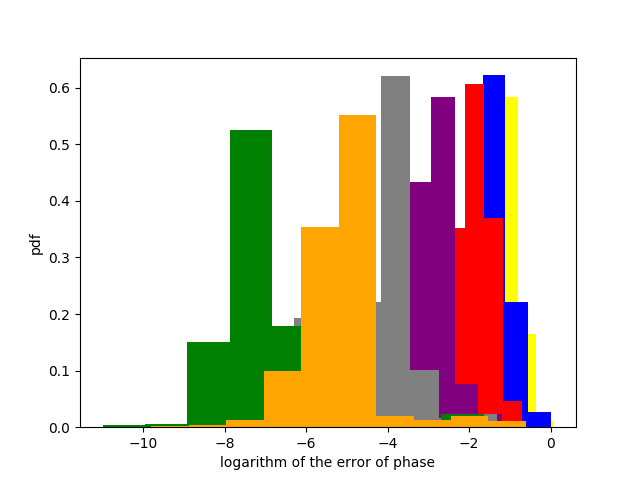}   
}  
\caption{PDF of the $\log_{10}(\Delta \Phi (\theta ))$ vs SNR(yellow space, blue space, red space, purple space, black space, orange space, green space individually represents SNR = -20, -10, 0, 20, 40, 60, 120).}  \label{fig5} 
\end{figure}

\subsection{The performance of the algorithm in theory}
From the previous section, we can get the performance of the spectrum reconstruction in one iteration. In one location, it needs three steps. Step 1: Obtain $\Phi (\frac{\hat{y}_{\tau_1}[i]}{\hat{y}_{\tau_2}[i]})$; it costs 2 runtime. Step 2: Obtain $l_r$ through the Eq. $\text{argmin} |(u_{\text{min}}+(l_r+0.5) r)\tau$ $  \text{mod}N \frac{2\pi }{N}-\Phi (\frac{\hat{y}_{\tau_1}[i]}{\hat{y}_{\tau_2}[i]}) |$; it costs 2 runtime. Step 3: Obtain new region $[u_{\text{min}},u_{\text{max}} )$, it costs 4 runtime. It costs 8 runtime in one location, so it totally costs $8R=O(\text{log}_{(l/(q+1))}L_m)$ runtime in a complete location.

As to every iteration, we can see the performance as follows. In the first iteration, it costs $R_1(w_1+B_1\text{log}B_1)+T_1$ runtime and $R_1w_1$ samples to find at least $K/2$ true frequency. In the second iteration, it costs $R_2(w_2+B_2\text{log}B_2+K-K_2)+T_2$ runtime and $R_2w_2$ samples to find at least $K/4$ true frequency. In the third iteration, it costs $R_3(w_3+B_3\text{log}B_3+K-K_3)+T_3$ runtime and $R_3w_3$ samples to find at least K/8 true frequency$\dots$. Suppose $w_1=B_1\text{log}(N/\sigma), K_1=K, B_1=O(K), w_2=B_2\text{log}(N/\sigma)$,$ K_2=K/2, B_2 =B_1/2$, $w_3=B_3\text{log}(N/\sigma),  K_3=K/4, B_3 =B_2/2, R_1=\text{log}_{(l/(q+1))}L_1, R_2=\text{log}_{(l/(q+1))}L_2, R_3=\text{log}_{(l/(q+1))}L_3$, we can get Lemma 2.

\begin{lemma} 
In the sFFT4.0 algorithm, it costs $O(\text{log}_{(l/(q+1))}(N/K) \ (K\text{log}N))$ runtime and $O(\text{log}_{(l/(q+1))}(N/K) \ (K\text{log}N))$ samples
\end{lemma} 
\begin{proof}
Analysis of runtime complexity:
\begin{equation*}
\begin{split}
& R_1(w_1+B_1\text{log}B_1)+T_1+R_2(w_2+B_2\text{log}B_2+K-K_2)+T_2+\\
& R_3(w_3+B_3\text{log}B_3+K-K_3)+T_3+ \dots\\
=& O( R_1(B_1\text{log}N)+R_2(B_2\text{log}N)+R_3(B_3\text{log}N)+...+ R_2 K/2+ R_3 3K/4+\dots) \\
\leq  &  O(R_1(B_1\text{log}N+(B_1\text{log}N)/2+(B_1\text{log}N)/4+…+K\text{log}K))  \\
= & O(\text{log}_{(l/(q+1))}(N/K) \ (K\text{log}N)) 
\end{split} 
\end{equation*}
Analysis of sampling complexity:
\begin{equation*}
\begin{split}
      & R_1w_1+R_2w_2+R_3w_3+ \dots \\
=  & O(R_1(B_1\text{log}N+(B_1\text{log}N)/2+(B_1\text{log}N)/4+\dots)) \\
= & O(\text{log}_{(l/(q+1))}(N/K) \ (K\text{log}N))  
\end{split} 
\end{equation*}
\end{proof}

Through Lemma 2, we find when $l/(q + 1)$ is large, the time complexity and sampling complexity are sub-linear correlation with $N$ even lower than $\text{log}N\text{log}N$.

\subsection{The comparison with other algorithms}
After analyzing the sFFT4.0 algorithm, the comparison with other algorithms is necessary. The performance of the sFFT1.0, sFFT2.0 and sFFT3.0 algorithm can be seen from \cite{IEEEexample:Hassanieh2012}, \cite{IEEEexample:Hassanieh2012-2}. The performance of the MPFFT algorithm can be seen from \cite{IEEEexample:Chiu2014}. The performance of the sFFT-DT1.0, sFFT-DT2.0 algorithm can be seen from \cite{IEEEexample:Hsieh2013}, \cite{IEEEexample:Hsieh2015}. The performance of the FFAST, R-FFAST algorithm can be seen from \cite{IEEEexample:Pawar2013}, \cite{IEEEexample:Pawar2018}, \cite{IEEEexample:Pawar2015}, \cite{IEEEexample:Ong2019}. The performance of the AAFFT algorithm can be seen from \cite{IEEEexample:Gilbert2002Near}, \cite{IEEEexample:Iwen2007Empirical}, \cite{IEEEexample:Gilbert2008A}. The performance of the fftw algorithm is common. The codes of the sFFT1.0, sFFT2.0\footnote{The code is available at http://groups.csail.mit.edu/netmit/sFFT/.}, sFFT3.0, MPFFT\footnote{The code is available at https://github.com/urrfinjuss/mpfft.}, sFFT-DT2.0\footnote{The code is available at https://www.iis.sinica.edu.tw/pages/lcs.}, R-FFAST\footnote{The code is available at https://github.com/UCBASiCS/FFAST.}, AAFFT\footnote{The code is available at https://sourceforge.net/projects/aafftannarborfa/.}, fftw\footnote{The code is available at http://www.fftw.org/.} algorithm are already open sources. Table 1 can be concluded with the information of all typical sFFT algorithms and fftw algorithm in theory.

\begin{table}[H] 
\caption{The performance of fftw algorithm and sFFT algorithms in theory}
\label{tab:2}
\setlength{\tabcolsep}{4pt}
\begin{tabular}{|p{50pt}|p{120pt}|p{110pt}|p{35pt}|}
\hline
algorithm& 
runtime complexity& 
sampling complexity & robustness \\
\hline
sFFT1.0&
$O(K^{\frac{1}{2}}N^{\frac{1}{2}}\ \text{log} ^{\frac{3}{2}}N)$& 
$N\left(1-\left(\frac{N-w}{N}\right)^{\ \text{log} N}\right)$& 
medium \\

sFFT2.0&
$O(K^{\frac{2}{3}}N^{\frac{1}{3}} \ \text{log} ^{\frac{4}{3}} N)$& 
$N\left(1-\left(\frac{N-w}{N}\right)^{\ \text{log} N}\right)$& 
medium  \\

sFFT3.0&
$O(K \ \text{log} N) $& 
$O(K \ \text{log} N) $& 
none  \\

sFFT4.0&
$O(K \ \text{log} N \ \text{log}_{(l/(q+1))}(N/K) )$& 
$O(K \ \text{log} N \ \text{log}_{(l/(q+1))}(N/K) )$& 
bad  \\

MPFFT&
$O(K \log N \ \text{log} _2 (N/K) )$& 
$O(K \log N \ \text{log} _2 (N/K) )$& 
good  \\

sFFT-DT1.0&
$O(K \ \text{log} K )$& 
$O(K)$& 
none  \\

sFFT-DT2.0&
$O(K \ \text{log} K + N)$& 
$O(K )$& 
medium  \\

FFAST&
$O( K \ \text{log} K) $& 
$O( K ) $& 
none  \\

R-FFAST&
$ O (K \ \text{log} ^{7/3} N)$& 
$O( K \ \text{log} ^{4/3} K) $& 
good  \\

AAFFT&
$O( K \text{poly} ( \ \text{log} N))$& 
$O( K \text{poly} ( \ \text{log} N))$& 
medium   \\

fftw&
$ O( N \ \text{log} N)$& 
$O(N)$& 
good  \\

\hline

\end{tabular}
\label{tab1}
\end{table}

	From the table, we can see the advantages of the sFFT4.0 algorithm is it has low runtime and sampling complexity of all sFFT algorithms except special condition algorithm for the order of magnitude of $N$. The disadvantages of the sFFT4.0 algorithm is it only has medium robustness under some parameters.

\section{Experimental evaluation}
In this section, we evaluate the performance of the sFFT4.0 algorithm. At first, we compare the algorithms' runtime, percentage of signal sampled and robustness characteristics with different parameters. Then we compare the algorithms' characteristics with similar sFFT algorithms using the same flat filter, including sFFT1.0, sFFT2.0 and MPSFT algorithm. At last, we compare the algorithm with other types of sFFT algorithms, including fftw, AAFFT, sFFT-DT, R-FFAST algorithm. All experiments are run on a Linux CentOS computer with 4 Intel(R) Core(TM) i5 CPU and 8 GB of RAM. 

\subsection{Experimental Setup}
In the experiment, the test signals are gained in a way that $K$ frequencies are randomly selected from $N$ frequencies and assigned a magnitude of 1 and a uniformly random phase. The rest frequencies are set to zero in the exact case or combined with additive white Gaussian noise in the general case, whose variance varies depending on the SNR required. The parameters of these algorithms are chosen so that they can make a balance between time efficiency and robustness. The general sparse case means SNR=20db. The new testing platform is developed from the old platform\footnote{https://github.com/ludwigschmidt/sft-experiments}. The detail of codes, data, report are all open sources\footnote{https://github.com/zkjiang/-/tree/master/docs}.

\subsection{The experiments with different parameters}
We plot Fig. \ref{fig6} representing run times vs Signal Size of the sFFT4.0 algorithm for different $l$ with $q=1$ and for different $q$ with $l=16$ in the general sparse case. From Fig. \ref{fig6}, we can see the runtime complexity is in direct ratio to $l$ and in inverse ratio to $q$, whether $N$ is big or small. 
\begin{figure} [H]\centering  
\subfigure[Different $l$ with $q=1$.] { 
\includegraphics[width=0.45\columnwidth]{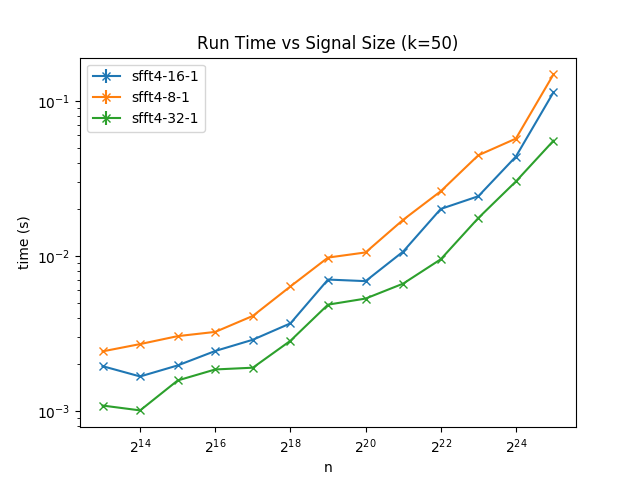}
}   
\subfigure[Different $q$ with $l=16$.] { 
\includegraphics[width=0.45\columnwidth]{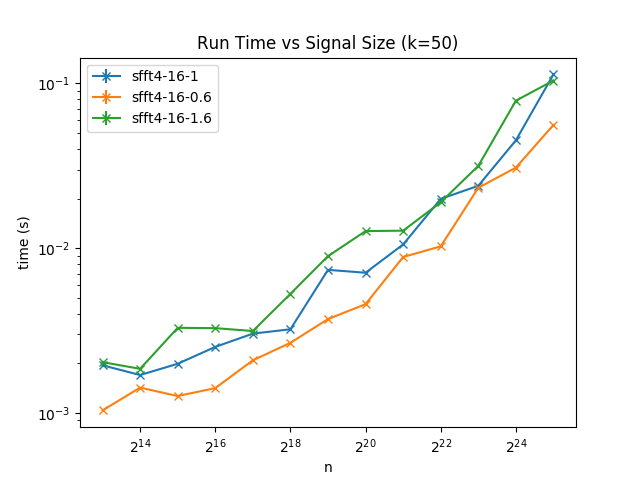}   
}   
\caption{Run time vs Signal Size of the sFFT4.0 algorithm for different parameters.}  \label{fig6} 
\end{figure}

We plot Fig. \ref{fig7} representing $L1$-error vs SNR of the sFFT4.0 algorithm for different $l$ with $q=1$ and for different $q$ with $l=16$. From Fig. \ref{fig7}, we can see the robustness is in direct ratio to $q$ and in inverse ratio to $l$, whether SNR is big or small. Considering comprehensively, we choose $l=16, q=1$ as the parameter of the algorithm, because it can make the algorithm have certain robustness and good runtime and sampling complexity.
\begin{figure} [H]\centering  
\subfigure[Different $l$ with $q=1$.] { 
\includegraphics[width=0.4\columnwidth]{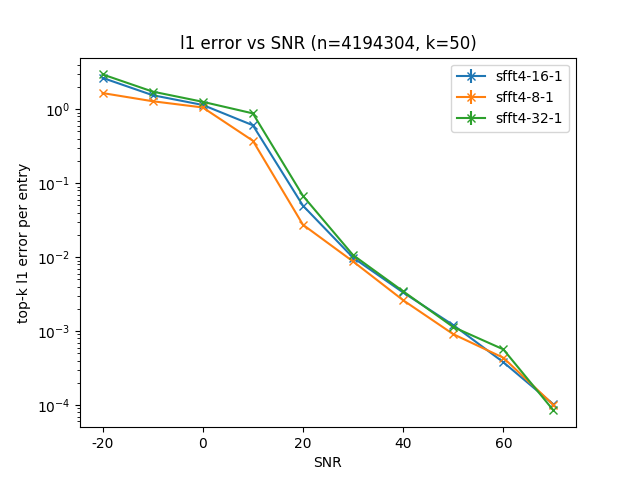}
}   
\subfigure[Different $q$ with $l=16$.] { 
\includegraphics[width=0.4\columnwidth]{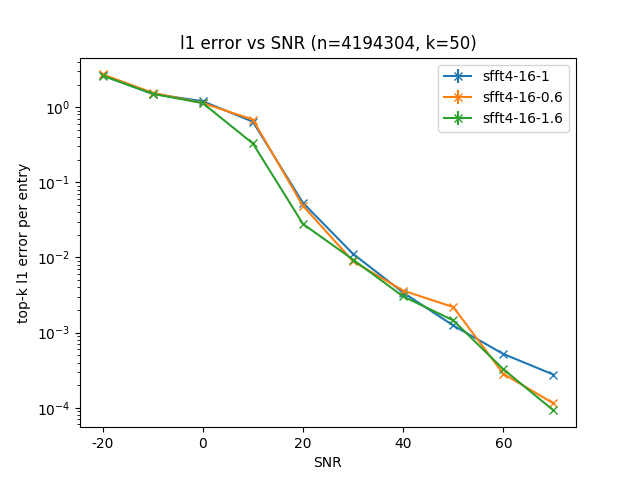}   
}   
\caption{$L1$-error vs SNR of the sFFT4.0 algorithm for different parameters.}  \label{fig7} 
\end{figure}

\subsection{The experiments with similar sFFT algorithms using the same flat filter}
We plot Fig. \ref{fig8} representing run times vs Signal Size and vs Signal Sparsity for the sFFT1.0, sFFT2.0, sFFT4.0, and MPSFT algorithm in the general sparse case. From Fig. \ref{fig8}, we can see 1)The run time of these four algorithms are approximately linear in the log scale as a function of $N$ and in the standard scale as a function of $K$. 2) Results of ranking the runtime of four algorithms is sFFT2.0 $>$ sFFT4.0 $>$ sFFT1.0 $>$ MPSFT when $N$ is large. 3) Results of ranking the runtime of four algorithms is sFFT4.0 $>$ sFFT2.0 $>$ sFFT1.0 $>$ MPSFT when $K$ is large. In a word, compared with other three algorithms, the sFFT4.0 algorithm has excellent runtime complexity.
\begin{figure} [H]\centering  
\subfigure[Run time vs signal size. ] { 
\includegraphics[width=0.4\columnwidth]{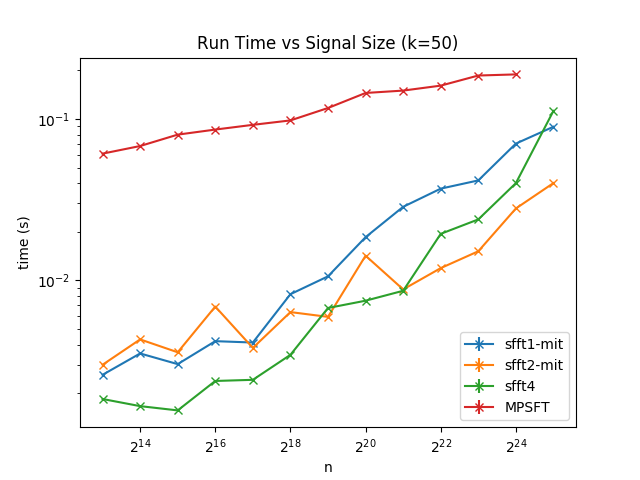}
}   
\subfigure[Run time vs signal sparsity.] { 
\includegraphics[width=0.4\columnwidth]{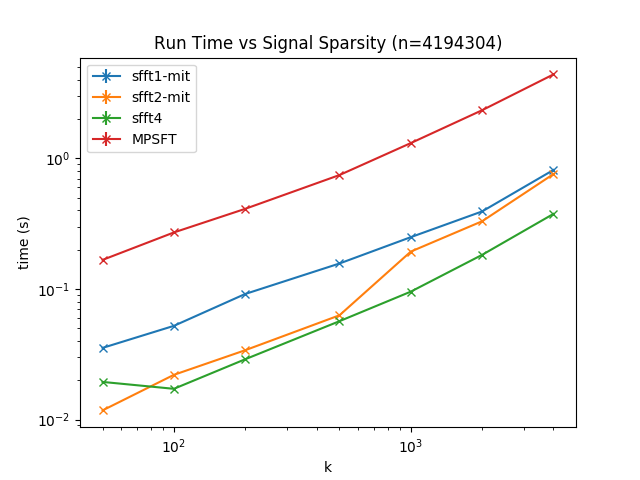}   
}   
\caption{Runtime of the sFFT4.0 algorithm and three similar algorithms in the general sparse case.}  \label{fig8} 
\end{figure}

We plot Fig. \ref{fig9} representing the percentage of the signal sampled vs signal size and vs signal sparsity for the sFFT1.0, sFFT2.0, sFFT4.0, and MPSFT algorithm in the general sparse case. From Fig. \ref{fig9}, we can see 1)The percentage of the signal sampled of these four algorithms are approximately linear in the log scale as a function of $N$ and in the standard scale as a function of $K$. 2)Results of ranking the sampling complexity of four algorithms is sFFT4.0 $>$ MPSFT $>$ sFFT2.0 $>$ sFFT1.0 because of different rounds in the algorithm. So compared with other three algorithms, the sFFT4.0 algorithm has excellent sampling complexity.
\begin{figure} [H]\centering  
\subfigure[Percentage of the signal sampled vs signal size. ] { 
\includegraphics[width=0.4\columnwidth]{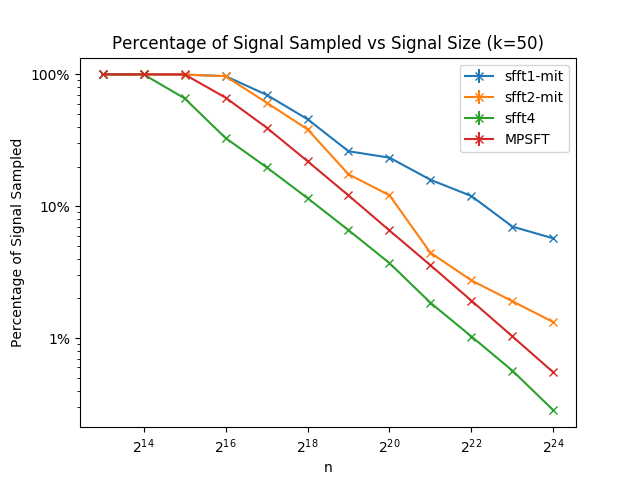}
}   
\subfigure[Percentage of the signal sampled vs sparsity.] { 
\includegraphics[width=0.4\columnwidth]{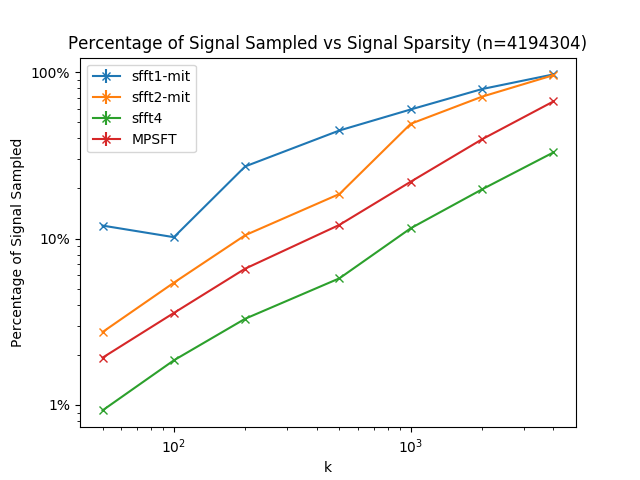}   
}   
\caption{Percentage of the signal sampled of the sFFT4.0 algorithm and three similar algorithms in the general sparse case.}  \label{fig9} 
\end{figure}

We plot Fig. \ref{fig10} representing the runtime and $L1$-error vs SNR for the sFFT1.0, sFFT2.0, sFFT4.0, and MPSFT algorithm. From Fig. \ref{fig10}, we can see 1)The runtime is approximately equal vs SNR. 2)To a certain extent, these four algorithms are all robust. 3)Results of ranking the robustness of four algorithms is MPSFT $>$ sFFT1.0 $>$ sFFT2.0 $>$ sFFT4.0. In a word, the sFFT4.0 algorithm has some robustness, but it is worse than other three algorithms. 
\begin{figure} [H]\centering  
\subfigure[Runtime vs SNR.] { 
\includegraphics[width=0.4\columnwidth]{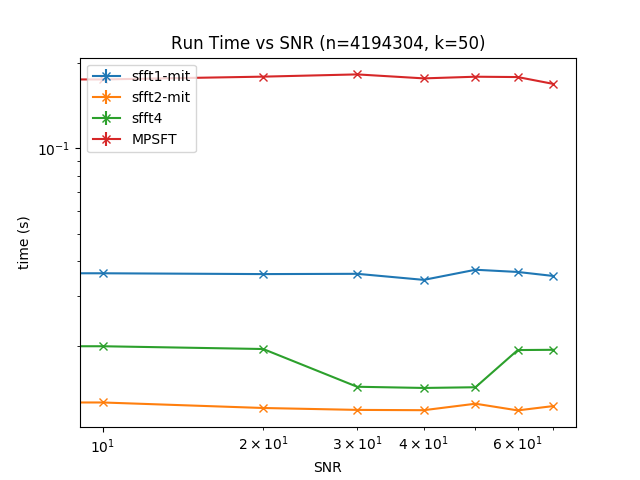}
}   
\subfigure[$L1$-error vs SNR.] { 
\includegraphics[width=0.4\columnwidth]{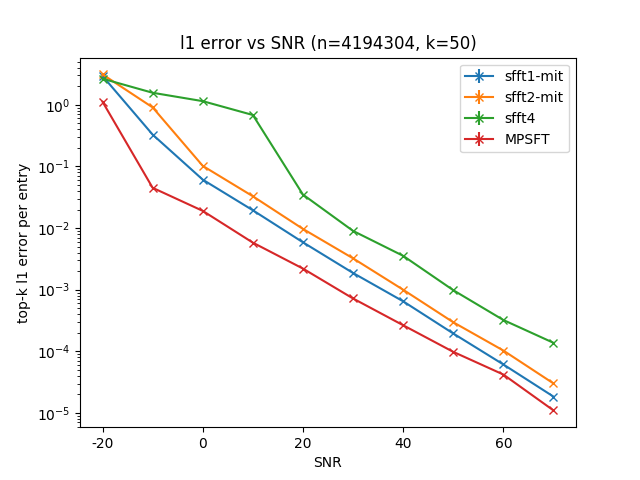}   
}   
\caption{Runtime and $L1$-error of the sFFT4.0 algorithm and three similar algorithms vs SNR.}  \label{fig10} 
\end{figure}

\subsection{The experiments with other types of sFFT algorithms}
We plot Fig. \ref{fig11} representing run times vs signal size and vs signal sparsity for sFFT4.0, AAFFT, R-FFAST, SFFT-DT, and fftw algorithm in the general sparse case. From Fig. \ref{fig11}, we can see 1)These five algorithms are approximately linear in the log scale as a function of $N$ except the fftw algorithm. These five algorithms are approximately linear in the standard scale as a function of $K$ except the fftw and SFFT-DT algorithm. 2)Results of ranking the runtime complexity of these five algorithms is sFFT4.0 $>$ AAFFT $>$ SFFT-DT $>$ fftw $>$ R-FFAST when $N$ is large. 3) Results of ranking the runtime complexity of these five algorithms is fftw $>$ SFFT-DT $>$ sFFT4.0 $>$ AAFFT $>$ R-FFAST when $K$ is large. In a word, compared with other types of algorithms, the sFFT4.0 algorithm has excellent runtime complexity that ten times better than the fftw algorithm.
\begin{figure} [H]\centering  
\subfigure[Run time vs signal size.] { 
\includegraphics[width=0.38\columnwidth]{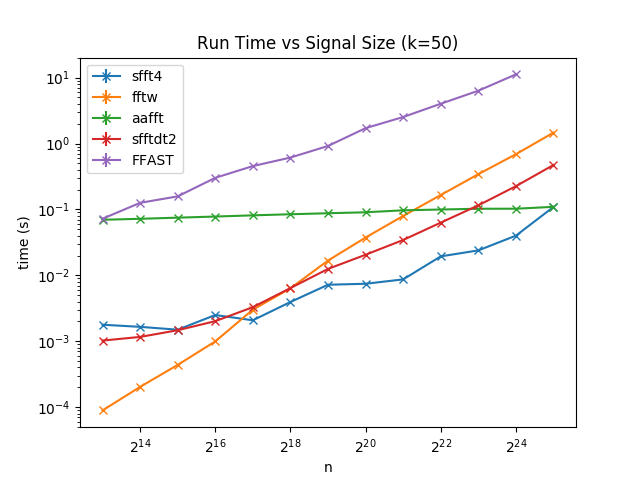}
}   
\subfigure[Run time vs signal sparsity.] { 
\includegraphics[width=0.38\columnwidth]{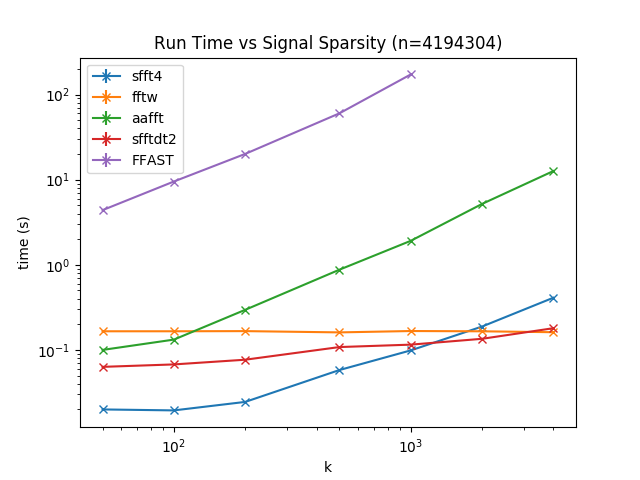}   
}   
\caption{Runtime of the sFFT4.0 algorithm and other four algorithms in the general sparse case.}  \label{fig11} 
\end{figure}

We plot Fig. \ref{fig12} representing the percentage of the signal sampled vs signal size and vs signal sparsity for sFFT4.0, AAFFT, R-FFAST, SFFT-DT and fftw algorithm in the general sparse case. From Fig. \ref{fig12}, we can see 1)These algorithms are approximately linear in the log scale as a function of $N$ except the fftw and SFFT-DT algorithm. These algorithms are approximately linear in the standard scale as a function of $K$ except the R-FFAST and SFFT-DT algorithm. 2)Results of ranking the sampling complexity of these five algorithms is R-FFAST $>$ sFFT4.0 $>$ AAFFT $>$ SFFT-DT $>$ fftw when $N$ is large. 3)Results of ranking the sampling complexity is SFFT-DT $>$ sFFT4.0 $>$ AAFFT $>$ fftw $>$ R-FFAST when $K$ is large. In a word, compared with other types of algorithms, the sFFT4.0 algorithm has excellent sampling complexity that one hundred times better than the fftw algorithm.
\begin{figure} [H]\centering  
\subfigure[Percentage of the signal sampled vs signal size.] { 
\includegraphics[width=0.38\columnwidth]{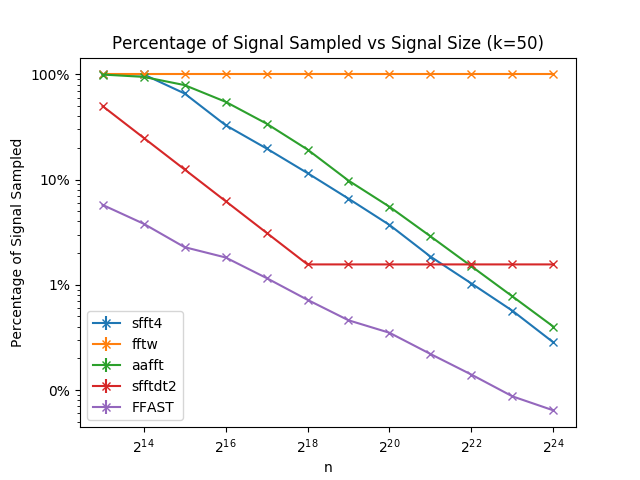}
}   
\subfigure[Percentage of the signal sampled vs signal sparsity.] { 
\includegraphics[width=0.38\columnwidth]{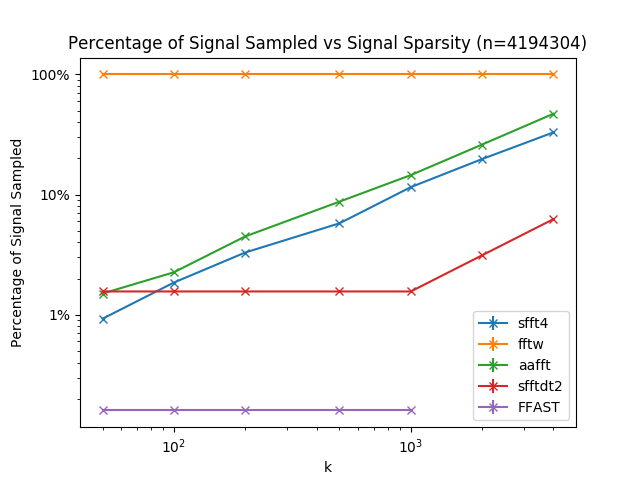}   
}   
\caption{Percentage of the signal sampled of the sFFT4.0 algorithm and other four algorithms in the general sparse case.}  \label{fig12} 
\end{figure}

We plot Fig. \ref{fig13} representing runtime and $L1$-error vs SNR for the sFFT4.0, AAFFT, SFFT-DT and fftw algorithm. From Fig. \ref{fig13} we can see 1)The runtime is approximately equal vs SNR. 2)To a certain extent, these five algorithms are all robust. 3)Results of ranking the robustness of four algorithms is fftw $>$ R-FFAST $>$ SFFT-DT $>$ AAFFT $>$ sFFT4.0. In a word, the sFFT4.0 algorithm has some robustness, but it is worse than other four algorithms. And only when SNR is bigger than 10db, the sFFT4.0 algorithm can deal with the noise interference.
\begin{figure} [H]\centering  
\subfigure[Runtime vs SNR.] { 
\includegraphics[width=0.4\columnwidth]{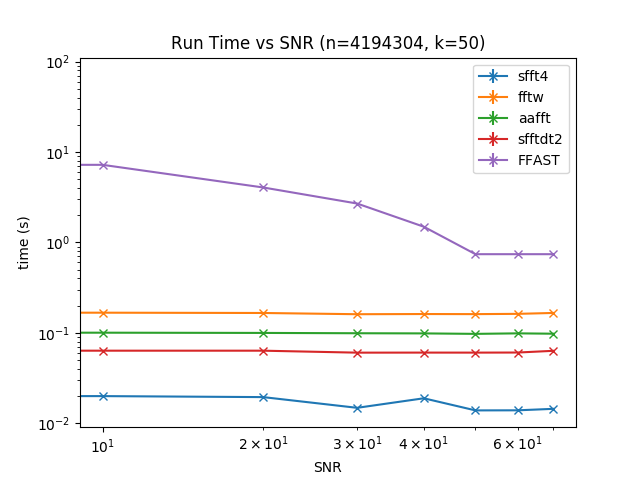}
}   
\subfigure[$L1$-error vs SNR.] { 
\includegraphics[width=0.4\columnwidth]{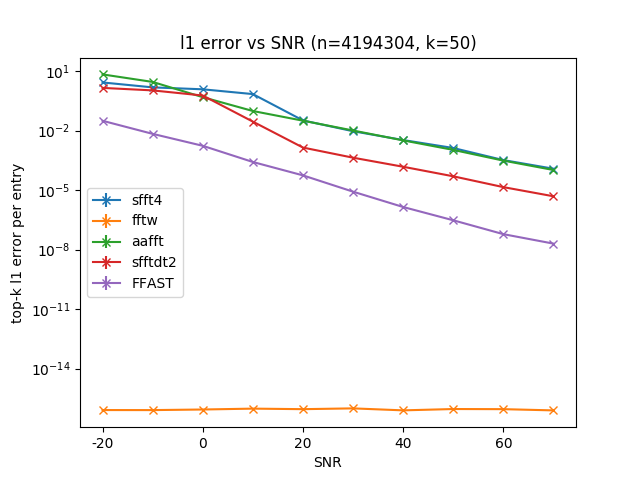}   
}   
\caption{Runtime and $L1$-error of the sFFT4.0 algorithm and other four algorithms vs SNR.}  \label{fig13} 
\end{figure}

\section{Conclusion}
In the first part, the paper provides a brief overview of the techniques used in sFFT algorithms including random spectrum permutation, window function and frequency bucketization. In the second part, we analyze the multiscale Sparse Fast Fourier Transform Algorithm(sFFT4.0 algorithm) in detail from four aspects: the overall flow of the algorithm, the steps of one iteration, the performance of the algorithm in theory, the comparison with other algorithms in theory. We get the conclusion of the performance of the sFFT4.0 algorithms and other comparison algorithms, including runtime complexity, sampling complexity and robustness in theory in Table 1. In the third part, we make three types of experiments for computing the signals of different SNR, different $N$, and different $K$ by a standard testing platform through the sFFT4.0 algorithm with different parameters, three similar algorithms, other four different algorithms and record the runtime, the percentage of the signal sampled and $L0, L1, L2$ error in every in the general sparse case. The analysis of the experiments satisfies theoretical inference.

The main contribution of this paper is 1)The sFFT4.0 algorithm using the multiscale approach method is analyzed in detail and implemented completely. 2)Develop a standard testing platform which can test more than eight typical sFFT algorithms under all kinds of signal on the basis of the old platform. 3)Get a conclusion of the character and performance of the sFFT4.0 algorithm in theory and practice. It has excellent runtime and sampling complexity that ten to one hundred times better than the fftw algorithm, although the robustness of the algorithm is medium.


%
 \section*{Conflict of interest}
 The authors declare that they have no conflict of interest.

 \section*{Open Access}
 
 This article is distributed under the terms of the Creative Commons Attribution 4.0 Interna-
tional License (http://creativecommons.org/licenses/by/4.0/), which permits unrestricted use, distribution, and reproduction in any medium, provided you give appropriate credit to the original author(s) and the source, provide a link to the Creative Commons license, and indicate if changes were made.


\bibliographystyle{IEEEtran}
\bibliography{IEEEabrv,IEEEexample}

\end{document}